\documentclass[11pt,reqno]{amsart}
\usepackage{amsaddr}
\usepackage[left=2cm,right=2cm,top=2cm,bottom=2cm]{geometry}
\usepackage[usenames,dvipsnames]{xcolor}
\usepackage{colortbl}
\usepackage[shortlabels]{enumitem}
\usepackage{amsmath,amsthm,amssymb,bm,bbm,dsfont}
\usepackage{thmtools}
\usepackage{thm-restate}

\usepackage{mathtools}\mathtoolsset{centercolon}
\mathtoolsset{showonlyrefs}
\usepackage[mathscr]{eucal}
\usepackage{upgreek}
\usepackage{setspace}
\usepackage{graphicx}
\usepackage{verbatim}
\usepackage{float}
\usepackage{placeins}
\usepackage{array}
\usepackage{booktabs}

\usepackage{threeparttable}
\usepackage[update,prepend]{epstopdf}
\usepackage{multirow}
\usepackage{amsfonts,amssymb,dsfont}
\usepackage[abs]{overpic}

\usepackage{tikz}
\usepackage{tikzscale}
\usetikzlibrary{calc}
\usetikzlibrary{patterns}
\usetikzlibrary{decorations.pathreplacing}
\usepackage{blox}

\bXLineStyle{-}

\definecolor{dullmagenta}{rgb}{0.4,0,0.4}   
\definecolor{darkblue}{rgb}{0,0,0.4}
\usepackage{hyperref}
\hypersetup{
	unicode=false,          
	pdftoolbar=true,        
	pdfmenubar=true,        
	pdffitwindow=false,     
	pdfstartview={FitH},    
	pdftitle={My title},    
	pdfauthor={Author},     
	pdfsubject={Subject},   
	pdfcreator={Creator},   
	pdfproducer={Producer}, 
	pdfkeywords={keyword1} {key2} {key3}, 
	pdfnewwindow=true,      
	linktocpage=true,
	colorlinks=true,        
	linkcolor=Red,          
	citecolor=ForestGreen,  
	filecolor=Magenta,      
	urlcolor=BlueViolet,    
}

\usepackage{doi}
\usepackage{caption, subcaption}
\usepackage{enumitem}
\usepackage{shadethm}

\makeatletter
\newcommand{\opnorm}{\@ifstar\@opnorms\@opnorm}
\newcommand{\@opnorms}[1]{%
	\left|\mkern-1.5mu\left|\mkern-1.5mu\left|
	#1
	\right|\mkern-1.5mu\right|\mkern-1.5mu\right|
}
\newcommand{\@opnorm}[2][]{%
	\mathopen{#1|\mkern-1.5mu#1|\mkern-1.5mu#1|}
	#2
	\mathclose{#1|\mkern-1.5mu#1|\mkern-1.5mu#1|}
}
\makeatother

\makeatletter
\ifx\@NODS\undefined%

\else%
\fi
\makeatother

\makeatletter
\ifx\@NODS\undefined%

\let\mathbb=\mathds
\else%
\fi
\makeatother

\usepackage[utf8]{inputenc}

\usepackage{xargs}
\usepackage[prependcaption]{todonotes}
\newcommandx{\eric}[2][1=]{\todo[inline, author={Eric}, linecolor=yellow,backgroundcolor=yellow!25,bordercolor=yellow,#1]{#2}}

\newcommandx{\ericside}[2][1=]{\todo[author={Eric}, linecolor=yellow,backgroundcolor=yellow!25,bordercolor=yellow,#1]{#2}}

\DeclareMathOperator*{\argmin}{\arg\min}

\DeclareMathOperator{\Tr}{Tr}
\DeclareMathOperator{\tr}{Tr}
\DeclareMathOperator{\e}{\mathrm{e}}

\newcommand{\be}{{\mathbf e}}

\newcommand{\cH}{{\mathcal{H}}}
\newcommand{\cP}{{\mathcal{P}}}

\newcommand{\cD}{{\mathcal{D}}}

\newcommand{\id}{{\mathrm{id}}}
\newcommand{\cX}{{\mathcal{X}}}
\newcommand{\cW}{{\mathcal{W}}}
\newcommand{\cB}{{\mathcal{B}}}
\newcommand{\cU}{{\mathcal{U}}}
\newcommand{\cK}{{\mathcal{K}}}
\newcommand{\cL}{{\mathcal{L}}}

\newcommand{\cF}{{\mathcal{F}}}

\def\0{{\mathbf{0}}}
\def\1{{\mathbf{1}}}
\def\2{{\mathbf{2}}}
\def\3{{\mathbf{3}}}
\def\4{{\mathbf{4}}}
\def\5{{\mathbf{5}}}
\def\6{{\mathbf{6}}}

\def\7{{\mathbf{7}}}
\def\8{{\mathbf{8}}}
\def\9{{\mathbf{9}}}

\def\be{\begin{equation}}
\def\ee{\end{equation}}
\def\bea{\begin{eqnarray}}
\def\eea{\end{eqnarray}}

\def\eps{\varepsilon}

\theoremstyle{plain}
\newtheorem{theo}{Theorem} 
\newtheorem{prop}[theo]{Proposition} 
\newtheorem{lemm}[theo]{Lemma} 
\newtheorem{coro}[theo]{Corollary} 
\newshadetheorem{conj}{Conjecture}

\theoremstyle{definition}

\theoremstyle{remark}
\newtheorem{remark}{Remark}[section]

\begin{document}
	
\let\origmaketitle\maketitle
\def\maketitle{
	\begingroup
	\def\uppercasenonmath##1{} 
	\let\MakeUppercase\relax 
	\origmaketitle
	\endgroup
}

\title{\bfseries \Large{ Strong converse bounds in quantum network information theory: distributed hypothesis testing and source coding     }}

\author{ {Hao-Chung Cheng$^{1}$, Nilanjana Datta$^{1}$, Cambyse Rouz\'e$^{1,2}$ }}
\address{\small  		
	$^{1}$Department of Applied Mathematics and Theoretical Physics, Centre for Mathematical Sciences\\University of Cambridge, Cambridge CB3 0WA, United Kingdom\\
	$^{2}$Technische Universit{\"a}t M{\"u}nchen, 80333 M{\"u}nchen, Germany}

\email{\href{mailto:HaoChung.Ch@gmail.com}{HaoChung.Ch@gmail.com}, 
\href{mailto:n.datta@statslab.cam.ac.uk}{n.datta@statslab.cam.ac.uk},
\href{mailto:rouzecambyse@gmail.com}{rouzecambyse@gmail.com}}

\date{\today}

\begin{abstract}
We consider a distributed quantum hypothesis testing problem with communication constraints, in which the two hypotheses correspond to two different states of a bipartite
quantum system, multiple identical copies of which are shared between Alice and Bob. They are allowed to perform local operations on their respective systems
and send quantum information to Charlie at limited rates. By doing measurements on the systems that he receives, Charlie needs to infer which of the two different states
the original bipartite state was in, that is, which of the two hypotheses is true. We prove that the Stein exponent for this problem is given by a regularized
quantum relative entropy. The latter reduces to a single letter formula when the alternative hypothesis consists of the products of the marginals of the null hypothesis,
and there is no rate constraint imposed on Bob. Our proof relies on certain properties of the  so-called quantum information bottleneck function.

The second part of this paper concerns the general problem of finding finite blocklength strong converse bounds in quantum network information theory. In the classical case, the analogue of this problem has been reformulated in terms of the so-called image size characterization problem. Here, we extend this problem to the classical-quantum setting and prove a second order strong converse bound for it. As a by-product, we obtain a similar bound for the Stein exponent for distributed hypothesis testing in the special case in which the bipartite system is a classical-quantum system, as well as 
for the task of quantum source coding with compressed classical side information. Our proofs use a recently developed tool from quantum functional inequalities, namely, the tensorization property of reverse hypercontractivity for the quantum depolarizing semigroup.
\end{abstract}
	
\maketitle

\section{Introduction}

Network information theory concerns the study of multi-user information theoretical tasks~\cite{el2011network}, which can be depicted using a network, which connects the different users and typically consists of multiple sources and
channels. Evaluating fundamental limits on the rate of information flow over the network, finding coding methods which achieve these limits, determining capacity regions, and establishing strong converse bounds, are the main problems
studied in this theory. These are challenging problems in the classical case itself, and even more so in the quantum setting. Examples of such tasks include, among others, distributed hypothesis testing, distributed compression and broadcasting. 
In this paper, we focus on the first two of these three examples, with the third being studied in a concurrent paper~\cite{CDR19b}. We first study the task of {\em{distributed hypothesis testing under communication constraints}}, which is elaborated in the next section. The key quantity that we focus on is the so-called Stein exponent. We obtain entropic expressions for it in different scenarios, employing properties of the quantum information bottleneck function~\cite{SCK+17,DHW18}. 

Next, we study the interesting problem of obtaining {\em{second-order strong converse bounds}}. 
For any given task in network information theory, the rates of information flow for which the task can be achieved with asymptotically vanishing error probability, defines a region, called the {\em{achievable rate region}}. A strong converse bound establishes that performing the task at rates lying outside this region leads to an error probability which goes to one in the asymptotic limit. A {\em{second-order}} strong converse bound additionally implies that the convergence to one is exponentially fast in the number of uses of the underlying resources.

We employ a powerful analytical toolkit analogous to the one used in the classical setting~\cite{liu2017beyond,LHV18} to obtain such bounds for a large class of network information tasks, including distributed hypothesis testing, source coding with side information at the decoder, and degraded broadcast channel coding. The method is based on an important functional inequality, namely the tensorization property of the reverse hypercontractivity of classical Markov semigroups. 
In \cite{BDR18}, the authors showed that the generalized quantum depolarizing semigroup satisfies such an inequality (see also \cite{Capel_2018}), and used it to prove finite blocklength, second-order strong converse bounds for the tasks of binary quantum hypothesis testing and point-to-point classical-quantum channel coding. The second half of this paper, which deals with strong converses, can be viewed as a continuation of the work initiated in \cite{BDR18} to the setting of classical-quantum network information theory. We focus, in particular, on the tasks of quantum distributed hypothesis testing with communication constraints and quantum source coding with classical side information. In a companion article \cite{CDR19b}, we show that this method also provides the second-order strong converse bounds for classical-quantum degraded broadcast channels.

\subsection{Distributed hypothesis testing under communication constraints}

 Hypothesis testing is a fundamental task, used for making statistical decisions about experimental data. It has been extensively studied in various seminal papers, including~\cite{NP33, Che52, Che56, Hoe65, Bla74}.
Usually the statistician has access to the entire data, and based on it, makes his inference as to which of a given set of hypotheses is true. Hypothesis testing is used widely and is of prime importance in Information Theory. The simplest version of it is that of binary hypothesis testing, in which the statistician needs to decide between two hypotheses -- the null hypothesis ($\mathsf{H}_0$) and the alternative hypothesis ($\mathsf{H}_1$). There is a tradeoff between the probabilities of the two possible errors that may be incurred: inferring the hypothesis to be $\mathsf{H}_1$ when $\mathsf{H}_0$ is true, or vice versa. These are called the {\emph{type I}} and {\emph{type II}} error probabilities, respectively. The celebrated Stein's lemma~\cite{Che56} provides an expression for the minimal type II error probability when the type I error probability is below a given threshold value.
\smallskip

A variant of the above task is one in which the data is distributed, e.g.~shared between two distant parties (say, Alice and Bob) who are not allowed to communicate with each other. The statistician (say, Charlie) does not have direct access to the data but instead learns about it from Alice and Bob, who can send information to him via noiseless classical channels at prescribed rates (say, $r_1$ and $r_2$, respectively). This variant of hypothesis testing is called {\emph{(bivariate) distributed hypothesis testing under communication constraints} \cite{AC86, Han87, HK89, Han98}.

In this paper, we study quantum versions of the above task. Suppose Alice and Bob share multiple (say $n$) identical copies of a bipartite quantum system $XY$, which is known to be in one of two states $\rho_{XY}$ and $\widetilde{\rho}_{XY}$. The system $X$ is with Alice and the system $Y$ is with Bob. They are allowed to perform local operations on their systems and then send the resulting systems to Charlie (via noiseless quantum channels). Charlie performs a joint measurement on the systems that he receives, in order to infer what the original state of the system $XY$ was. In the context of hypothesis testing, the two hypotheses are given by 
\begin{align} \label{eq:general}
\begin{dcases}
\mathsf{H}_0: \rho_{XY}^{\otimes n} \\
\mathsf{H}_1: \widetilde{\rho}_{XY}^{\otimes n} \\
\end{dcases} \quad {\hbox{with}}\,\, n\in\mathbb{N}.
\end{align}
The local operations that Alice and Bob do, on the systems $X^n$ and $Y^n$ in their possession, are given by linear, completely positive trace-preserving maps ({{i.e.~}}quantum channels) $\mathcal{F}_n\equiv\mathcal{F}_n^{X^n\to W^n}$ and $\mathcal{G}_n\equiv\mathcal{G}_n^{Y^n\to \widetilde{W}^n}$, respectively.
So the state that Charlie receives is one of the following two
\begin{align}
\sigma_{W^n \widetilde{W}^n} :&= \Big( \mathcal{F}_n^{X^n\to W^n} \otimes  \mathcal{G}_n^{Y^n\to \widetilde{W}^n} \Big) \rho_{XY}^{\otimes n}; \label{eq:sigma0} \\
\widetilde{\sigma}_{W^n \widetilde{W}^n}   :&= \Big(  \mathcal{F}_n^{X^n\to W^n} \otimes \mathcal{G}_n^{Y^n\to \widetilde{W}^n} \Big) \widetilde{\rho}_{XY}^{\otimes n}. \label{eq:sigma1}
\end{align}

He performs a binary POVM on the state that he receives, to decide which of the two hypotheses, $\mathsf{H}_0$ and $\mathsf{H}_1$, is true. We denote the probabilities of type-I and type-II errors associated to this hypothesis testing problem as follows:
$\alpha$ and $\beta$, respectively.}
In this paper, we are interested in the trade-off between the type-I and the type-II errors given that the communication from both Alice and Bob to Charlie are limited. This is the setting of \emph{(bivariate) distributed quantum hypothesis testing under communication constraints\footnote{Henceforth, we suppress the phrase {\emph{bivariate}} for simplicity.}.}
Specifically, we assume that
\begin{align}
\log |W^n| \leq n r_1; \quad\log |\widetilde{W}^n| \leq n r_2
\end{align}
for some $r_1, r_2 >0$
(see Figure~\ref{fig:protocol} below). Here, we use $|X|$ to denote the dimension of the Hilbert space $\cH_X$ corresponding to the system $X$. 
\medskip

The operational quantity that we focus on is the following: given any $\eps\in[0,1]$, $r_1,r_2>0$, we define the \emph{quantum Stein exponent} for this hypothesis testing task (hereafter referred to simply as the {\em{Stein exponent}}) as
\begin{align} \label{eq:q_Stein_exponent}
e\left(\eps|r_1, r_2 \right) := \liminf_{n\to\infty} \left\{-\frac1n \log \beta_{r_1,r_2}(n,\eps) \right\},
\end{align}
where
\begin{align} \label{eq:beta_R}
\beta_{r_1,r_2}(n,\eps) := \inf_{ \substack{\mathcal{F}_n:\frac1n\log|W^n|\leq  r_1\\ \mathcal{G}_n:\frac1n\log|\widetilde{W}^n| \leq r_2} } \beta(n,\eps,\mathcal{F}_n,\mathcal{G}_n),
\end{align}
and 
\begin{align}
\beta(n,\eps, \mathcal{F}_n, \mathcal{G}_n ) := \inf_{ \substack{0\leq T \leq \mathds{1} \\ \Tr\left[(\mathds{1}-T) \sigma_{W^n \widetilde{W}^n} \right] \leq \eps } } \Tr\left[T\, \widetilde{\sigma}_{W^n \widetilde{W}^n} \right].
\end{align}
denotes the the optimal type-II error given that the type-I error is at most $\eps$.

\begin{figure}[h!]
	\centering
	\includegraphics[width = 0.90\linewidth]{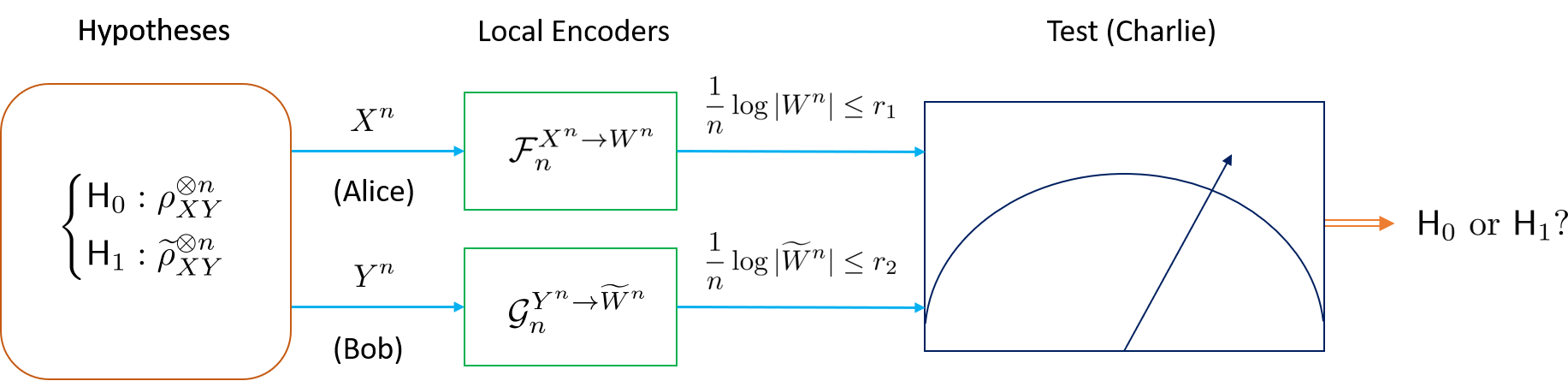}
	\caption{	The scheme of distributed quantum hypothesis testing. Moreover, unlimited entanglement is allowed between Alice and Charlie, and Bob and Charlie.
		\label{fig:protocol}}
\end{figure}
Note that in the absence of the communication constraints, {{i.e.}}~when $r_1$ and $r_2$ are both infinite, the Stein exponent, defined above, reduces to the usual quantum Stein exponent, which is known to be equal to the relative entropy between the two hypotheses by the quantum Stein lemma~\cite{HP91, NO00}. Hence,
\begin{align} \label{eq:qSL}
e(\eps |\infty, \infty) = D({\rho}_{XY}||\widetilde{\rho}_{XY}), \quad \forall \eps \in (0,1),
\end{align}
where  $D(\rho\|\sigma)$  denotes the quantum relative entropy \cite{Ume62} (see \eqref{eq:relative_entropy} in Section~\ref{sec:notation}).

Our results for the tasks of distributed hypothesis testing can be summarized as follows:
\begin{itemize}
\item{Firstly, we show that the Stein exponent is given by a regularized quantum relative entropy (Theorem~\ref{theo:Stein_general}):
\begin{align}
\lim_{\eps\to 0}e\left(\eps|r_1, r_2 \right) := \lim_{n\to \infty} \,\, \sup_{ \substack{\mathcal{F}_n:\frac1n\log|W^n|\leq  r_1\\ \mathcal{G}_n:\frac1n\log|\widetilde{W}^n| \leq r_2} } \frac1n D \left(\sigma_{W^n \widetilde{W}^n}\| \widetilde{\sigma}_{W^n \widetilde{W}^n} \right). \label{res-1}
\end{align}
Note that a single-letter expression for expression on the right-hand side of \eqref{res-1} is not even known in the classical case.
\smallskip

} 
\item{Secondly, we study the case in which the alternative hypothesis is a product state of the marginals, {i.e.}~$\widetilde{\rho}_{XY} = \rho_X\otimes \rho_Y$. We refer to this case as the \emph{distributed quantum hypothesis testing against independence} or simply \emph{testing against independence}.
We establish the following single-letter formula for the corresponding Stein exponent when there is no rate constraint imposed on Bob (Theorem~\ref{theo:Stein_independence}): setting $r_1\equiv r$ we have,
\begin{align}  \label{eq:single1}
\lim_{\eps\to 0} e(\eps| r,\infty) =  \sup_{ \substack{\mathcal{N}^{X\to U} \\ \frac12 I(X';U)_{\tilde{\tau}} \leq r}  } I\left( U;Y\right)_\omega,
\end{align}
where $\mathcal{N}^{X\to U}$ is a quantum channel from systems $X$ to $U$,
$$
\omega_{UYR} := (\mathcal{N}^{X\to U} \otimes \id_{YR}) \psi_{XYR},$$
 with $\psi_{XYR}$ being a purification of $\rho_{XY}$, $R$ being an inaccessible reference system;
$$\tilde{\tau}_{X'U} := (\id_{X'}\otimes \mathcal{N}^{X\to U}) \tau_{X'X},$$ where $\tau_{X'X}$ is a purification of $\rho_X$;
and $I(A;B)$ is the quantum mutual information of a state $\tau_{AB}$ (see \eqref{eq:mutual} in Section~\ref{sec:notation}). We remark that the above quantity is the dual of the so-called \emph{quantum information bottleneck function}~\cite{SCK+17,DHW18}. 
\medskip

\noindent
{\emph{Remark:}} When the rate $r$ is above the von Neumann entropy of $X$, the supremum in \eqref{eq:single1} is attained by an identity map. Hence, our result immediately yields that
\begin{align} 
\lim_{\eps\to 0} e(\eps| r,\infty) = I(X;Y)_\rho\equiv D(\rho_{XY}\|\rho_X\otimes \rho_Y)\,,
\end{align}
which coincides with the expression for usual quantum Stein exponent in the absence of communication constraints.
\smallskip

}
\item{ The statement of the original quantum Stein lemma~\cite{HP91, ON00} holds for all $\eps \in (0,1)$. This in turn implies the so-called \emph{strong converse property} for quantum hypothesis testing, i.e.~if for any test $T_n$, the type II error probability is restricted to be less than or equal to $\e^{-nr}$, with $r$ being greater the the relative entropy between the states corresponding to the two hypotheses, then the associated type I error goes to one as $n \to \infty$ (with $n$ being the number of copies of the states available). 

Our third result, given by Theorem~\ref{theo:sc_cq}, provides a second-order strong converse bound on the Stein exponent for the distributed quantum hypothesis testing task introduced in the previous section, {\emph{in the special case in which the system $X$ is classical}} (associated with a random variable taking values in a finite set $\cX$). It states that 
for every $\eps\in(0,1)$, there exists a $K>0$ such that for all $r>0$,
 \begin{align}\label{strongconverse}
 -\frac1n \log \beta_{r,\infty}(n,\eps) \leq  \sup_{ \substack{\mathcal{N}^{X\to U} \\  I(X';U)_\omega \leq r}  } I\left( U;Y\right)_\omega +
 \frac{K}{\sqrt{n}}\,,
 \end{align}
where 
$$
\omega_{UX'Y} := \left(\mathcal{N}^{X\to U} \otimes {\rm{id}}_{X'} \otimes {\rm{id}}_{Y}\right) \rho_{XX'Y}.
$$
 with $\mathcal{N}^{X\to U}$ being a classical channel (i.e.~a stochastic map)
that maps the random variable $X$ to $U$, and $X'$ being a copy of $X$. See Section~\ref{strongconverseQHT} for details.
}
\end{itemize}
\subsection{Image-size characterization problem and source coding with classical side information} 

The key ingredient of the proof of our result (Theorem~\ref{theo:sc_cq}) on the second-order strong converse bound (Proposition~\ref{prop:variational}) for distributed quantum hypothesis testing, can also be employed to establish strong converse bounds for a more general task, namely, the so-called \textit{image-size characterization problem}. The latter, which provides a unifying framework for the analysis of a wide variety of source and channel network problems, can be roughly explained as follows: let $\cX,\mathcal{Y}$ be two (say finite) sets, and provide $\mathcal{X}^n$, respectively $\mathcal{Y}$, with a non-negative measure $\mu_n$, respectively $\nu_{{Y}}$. Then, given any classical channel $Q_{Y|X}$ and $0<1\eps<1$, we are interested in a lower-bound on the $\nu_Y^{\otimes n}$-measure of any set $\mathcal{B}\subseteq \mathcal{Y}^n$ in terms of its $\eps$-\textit{preimage} under $Q_{Y^n|X^n}$: more precisely, given any $0<\eps,r<1$, we want to find a lower-bound on the following quantity:

\begin{align}\label{eq:ISC}
\min_{\mathcal{B}\subset \mathcal{Y}^n :\,\mu_n\big(\big\{x^n|\,\,Q_{Y^n|X^n=x^n}(\mathcal{B})>1-\eps)  \big\}\big)>r}\,\nu_Y^{\otimes n}(\mathcal{B})\,.
\end{align}

 In Theorem \ref{theoimagesize}, we extend the above definition to the classical-quantum setting, and provide a second order strong converse bound (similar to the one of (\ref{strongconverse})) on \eqref{eq:ISC}. Due to the wide use of the image-size characterization method in classical network information theory, we expect our result to find applications in corresponding classical-quantum settings. As a first application of this, we establish a strong converse 
bound for the task of quantum source coding with compressed classical side information at the decoder \cite{hsieh2016channel, DHW18} (see Figure~\ref{WAK} and Theorem~\ref{theosourcecoding}). 
\medskip
\begin{figure}[h!]
	\centering
	\includegraphics[width = 0.6\linewidth]{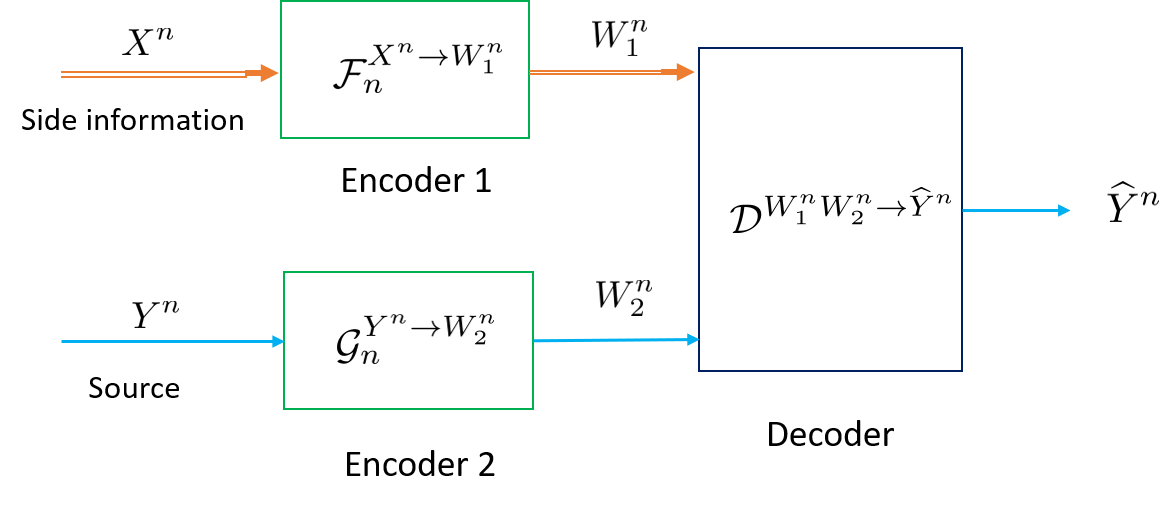}
	\caption{	The scheme of quantum source coding with classical side information.
		\label{WAK}}
\end{figure}

\textbf{Layout of the paper:} In Section~\ref{sec:notation} we introduce the necessary notations and definitions. In Section~\ref{sec:Stein}, we show that the Stein exponent for distributed quantum hypothesis testing under communication constraints, is given by a regularized formula involving the quantum relative entropy. In Section~\ref{sec:single-letter0}, we obtain a single-letter expression for the Stein exponent for the task of testing against independence.
In Section~\ref{sec:sc}, we prove a second-order strong converse bound for the tein exponent when the bipartite states corresponding to the two hypotheses are classical-quantum states. We also obtain similar bounds for the image-size characterization problem, and for the task of quantum source coding with classical side information. Lastly, we conclude this paper with a discussion of the results and future directions in section~\ref{sec:discussions}.

\section{Notations and Definitions} \label{sec:notation}
Throughout this paper, we consider finite-dimensional Hilbert spaces, and discrete random variables which take values in finite sets. The subscript of a Hilbert space (say $B$), denotes the quantum system (say $\cH_B$) to which it is associated. We denote its dimension as $d_B := {\rm{dim}}\,{\mathcal{H}}_B$. 
Let $\mathbb{N}$, $\mathbb{R}$, and $\mathbb{R}_{\geq 0}$ be the set of natural numbers, real number, and non-negative real numbers, respectively.
Let ${\mathcal{B}}({\mathcal{H}})$ denote the algebra of linear operators acting on a Hilbert space ${\mathcal{H}}$, ${\mathcal{P}}({\mathcal{H}}) \subset  {\mathcal{B}}({\mathcal{H}})$ denote the set of positive semi-definite operators, ${\mathcal{D}}({\mathcal{H}}) \subset {\mathcal{P}}({\mathcal{H}})$ the set of quantum states (or density matrices): ${\mathcal{D}}({\mathcal{H}}) :\{ \rho \in {\mathcal{P}}({\mathcal{H}})\,:\, \tr \rho = 1\}$. We will use the notation $\rho^{X}_{Y}$ for a state depending on a random variable $X$, where $Y$ denotes the  register corresponding to the Hilbert space on which the state acts. A quantum operation (or quantum channel) is a superoperator given by a linear completely positive trace-preserving (CPTP) map. A quantum operation ${\mathcal{N}}^{A \to B}$ maps operators in ${\mathcal{B}}({\mathcal{H}}_A)$ to operators in ${\mathcal{B}}({\mathcal{H}}_B)$. We denote the identity superoperator as ${\rm{id}}$. A superoperator $\Phi: {\mathcal{B}}({\mathcal{H}}) \to {\mathcal{B}}({\mathcal{H}})$ is said to be unital if $\Phi(\mathbb{I}) = \mathbb{I}$, where $\mathbb{I}$ denotes the identity operator in ${\mathcal{B}}({\mathcal{H}})$ .

The von Neumann entropy of a state $\rho$ is defined as $S(\rho):= - \tr [\rho \log \rho]$. Here, and henceforth, logarithms are taken to base $2$. 
The quantum relative entropy between a state $\rho \in {\mathcal{D}}({\mathcal{H}})$ and a positive semi-definite operator $\sigma$ is defined as
\begin{align} \label{eq:relative_entropy}
D(\rho||\sigma) &: = \Tr\left[\rho (\log \rho - \log \sigma)\right].
\end{align}
It is well-defined if ${\rm{supp}}\,\rho \subseteq {\rm{supp}}\,\sigma$, and is equal to $+\infty$ otherwise. Here ${\rm{supp}}\,A$ denotes the support of the operator $A$. 
The quantum relative R\'enyi entropy of order $\alpha$~\cite{Pet86} is defined for $\alpha \in (0,1)$ as follows:
\begin{align}
D_\alpha(\rho||\sigma) &:= \frac{1}{\alpha-1} \log \Tr\left( \rho^\alpha \sigma^{1-\alpha} \right).
\end{align}
It is known that $D_\alpha(\rho||\sigma) \to D(\rho||\sigma)$ as $\alpha \to 1$  (see e.g.~\cite[Corollary 4.3]{Tom16}, \cite{Ume62}).
An important property satisfied by these relative entropies is the so-called {\em{data-processing inequality}}, which is given by
$ D_\alpha(\Lambda(\rho)\|\Lambda(\sigma)) \leq D_\alpha(\rho\|\sigma)$ for all $\alpha\in(0,1)$ and quantum operations $\Lambda$.
This induces corresponding data-processing inequalities for the quantities derived from these relative entropies, such as the quantum mutual information information \eqref{eq:mutual} and the conditional entropy \eqref{eq:conditional}.

For a bipartite state $\rho_{AB} \in \mathcal{D}(\mathcal{H}_A \otimes \mathcal{H}_B)$, the quantum mutual information and the conditional entropy are given in terms of the quantum relative entropy as follows:
\begin{align}
I(A;B)_\rho &= D\left( \rho_{AB} \| \rho_A \otimes \rho_B \right); \label{eq:mutual} \\
H(A|B)_\rho &= -D\left( \rho_{AB} \| \mathbb{I}_A \otimes \rho_B \right). \label{eq:conditional}
\end{align}

\section{The Stein exponent for distributed quantum hypothesis testing under communication constraints}  \label{sec:Stein}

\subsection{The general situation}

In this section, we establish bounds on the Stein exponent for the hypotheses given in  ~\eqref{eq:general}.
We first introduce an entropic quantity:
\begin{align} 
\begin{split} \label{eq:theta}
\theta(r_1,r_2) &:= \sup_{n \in\mathbb{N}} \theta_n(r_1,r_2), \\
\theta_n(r_1,r_2) &:= \sup_{ \substack{\mathcal{F}_n:\frac1n\log|W^n|\leq  r_1\\ \mathcal{G}_n:\frac1n\log|\widetilde{W}^n| \leq r_2} } \frac1n D\left( \sigma_{W^n \widetilde{W}^n}\| \widetilde{\sigma}_{W^n \widetilde{W}^n} \right)\,,
\end{split}
\end{align}
where the states $ \sigma_{W^n \widetilde{W}^n}$ and $\widetilde{\sigma}_{W^n \widetilde{W}^n}$ are defined through equations (\ref{eq:sigma0}) and (\ref{eq:sigma0}).
We provide properties of $\theta$ below, which will be useful later.
\begin{lemm} \label{lemm:prop_theta}
	Let $\theta_n$ and $\theta$ be defined in  ~\eqref{eq:theta}. Then, the following holds:
	\begin{enumerate}[(a)]
		\item\label{prop_theta-a} For every $r_1,r_2\geq 0$, the map $n\mapsto \theta_n(r_1,r_2)$ is monotonically increasing. Hence,
		\begin{align}
		\theta(r_1,r_2) = \lim_{n\to \infty} \theta_n(r_1, r_2).
		\end{align}
		
		\item\label{prop_theta-b} The map $(r_1,r_2) \mapsto \theta(r_1,r_2)$ is jointly concave on $[0,\infty)\times[0,\infty)$.
		
		\item\label{prop_theta-c} The map $(r_1,r_2) \mapsto \theta(r_1,r_2)$ is continuous on $(0,\infty)\times (0,\infty)$.
	\end{enumerate}
\end{lemm}
\begin{proof}
	The proof follows the same reasoning as \cite[Lemma 1]{AC86}, and is provided here for sake of completeness.
	\begin{enumerate}
		\item[\ref{prop_theta-a}] 
	By definition, the map $n\mapsto n\,\theta_n(r_1,r_2)$ is super-additive, i.e.~
		\begin{align}
		&(k+\ell) \theta_{k+\ell}(r_1,r_2) \notag \\
		&= \sup_{\mathcal{F}_{k+\ell} ,\, \mathcal{G}_{k+\ell  } } \left\{  D\left( \sigma_{W^{k+\ell} \widetilde{W}^{k+\ell } }\| \widetilde{\sigma}_{W^{k+\ell} \widetilde{W}^{k+\ell} } \right): |W^{k+\ell}|\leq 2^{(k+\ell) r_1}, \, |\widetilde{W}^{k+\ell}| \leq2^{(k+\ell) r_2} \right\} \notag \\
		&\geq \sup_{ \substack{\mathcal{F}_{k}^{(1)}\otimes \mathcal{F}_\ell^{(2)} \\ \mathcal{G}_{k}^{(1)}\otimes \mathcal{G}_\ell^{(2)} } } \left\{  D\left( \sigma_{W^{k} \widetilde{W}^{k } } \otimes \sigma_{W^{\ell} \widetilde{W}^{\ell } } \| \widetilde{\sigma}_{W^{k} \widetilde{W}^{k} } \otimes  \widetilde{\sigma}_{W^{\ell} \widetilde{W}^{\ell} } \right):|W^{k}|\leq 2^{k r_1},\, |W^{\ell}|\leq 2^{\ell r_1}, \, |\widetilde{W}^{k}|\leq 2^{k r_2}, \, |\widetilde{W}^{\ell}|\leq 2^{\ell r_2} \right\} \notag \\
		&= \sup_{\mathcal{F}_{k} ,\, \mathcal{G}_{k  } } \left\{  D\left( \sigma_{W^{k} \widetilde{W}^{k} }\| \widetilde{\sigma}_{W^{k} \widetilde{W}^{k} } \right): |W^{k}|\leq 2^{k r_1}, \, |\widetilde{W}^{k}| \leq 2^{k r_2} \right\} \notag \\
		&\quad + \sup_{\mathcal{F}_{\ell} ,\, \mathcal{G}_{\ell} } \left\{  D\left( \sigma_{W^{\ell} \widetilde{W}^{\ell} }\| \widetilde{\sigma}_{W^{\ell} \widetilde{W}^{\ell} } \right): |W^{\ell}|\leq  2^{\ell r_1}, \, |\widetilde{W}^{\ell}| \leq 2^{\ell r_2} \right\} \notag \\
		&= k\theta_k(r_1,r_2) + \ell\theta_\ell(r_1,r_2).
		\end{align}
		Here, in the first inequality we restrict the constraint to local maps, and the second equality follows from the additivity of the relative entropy.
		Since this holds for every $k,\ell\in\mathbb{N}$, the assertions in \ref{prop_theta-a} hold.
		
		\item[\ref{prop_theta-b}] For every $n\in\mathbb{N}$, we have
		\begin{align}
		&\theta_{2n}\left( \frac{r_1 + \bar{r}_2}{2}, \frac{r_1 + \bar{r}_2}{2} \right)=
		\sup_{\mathcal{F}_{2n},\, \mathcal{G}_{2n} } \left\{ \frac{1}{2n} D\left( \sigma_{W^{2n} \widetilde{W}^{2n} }\| \widetilde{\sigma}_{W^{2n} \widetilde{W}^{2n} } \right): |W^{2n}|\leq   2^{n(r_1 + \bar{r}_1)}, \, |\widetilde{W}^{2n}| \leq 2^{n(r_2 + \bar{r}_2)} \right\} \notag \\
		&\geq \frac12 \sup_{ \substack{\mathcal{F}_{n}^{(1)}\otimes \mathcal{F}_n^{(2)} \\ \mathcal{G}_{n}^{(1)}\otimes \mathcal{G}_n^{(2)} } } \left\{  D\left( \sigma_{W^{n}_{(1)} \widetilde{W}^{n}_{(1)} } \otimes \sigma_{W^{n}_{(2)} \widetilde{W}^{n }_{(2)} } \| \widetilde{\sigma}_{W^{k}_{(1)} \widetilde{W}^{k}_{(2)} } \otimes  \widetilde{\sigma}_{W^{\ell}_{(1)} \widetilde{W}^{\ell}_{(2)} } \right): \right. 
		\notag \\ &\qquad\qquad\qquad\quad \left. |W^{n}_{(1)}|\leq 2^{n r_1},\, |W^{n}_{(2)}|\leq 2^{n \bar{r}_1}, \, |\widetilde{W}^{k}_{(1)}|\leq 2^{n r_2}, \, |\widetilde{W}^{n}_{(2)}|\leq 2^{n \bar{r}_2} \right\} \notag \\
		&= \frac12\left( \theta_n(r_1,r_2) + \theta(\bar{r}_1,\bar{r}_2) \right).
		\end{align}
		Then, letting $n\to \infty$ on both sides and recalling item~\ref{prop_theta-a} implies the assertion in item~\ref{prop_theta-b}.
		
		\item[\ref{prop_theta-c}] 
		Since a concave function is continuous in its interior, the assertion in item~\ref{prop_theta-c} follows from item~\ref{prop_theta-b}.
	\end{enumerate}
\end{proof}

\begin{theo}[The Stein exponent] \label{theo:Stein_general}
	For the hypotheses in  ~\eqref{eq:general}, the following holds for every $r_1, r_2 > 0$.
	\begin{itemize}
		\item[a)] Achievability (direct part):
		\begin{align} \label{eq:Stein_general0}
		e(\eps|r_1,r_2)\geq \theta(r_1, r_2), \quad \forall\,\eps\in(0,1);
		\end{align}
		
		\item[b)] Optimality (weak converse):
		\begin{align} \label{eq:Stein_general1}
		\lim_{\eps\to0} e(\eps|r_1,r_2)\ \leq \theta(r_1,r_2).
		\end{align}
	\end{itemize}
\end{theo}
\begin{proof}
\begin{itemize}
	\item[a)]
	We first prove the achievability part. Fix $\eps\in(0,1)$.
	Let us fix a $k\in\mathbb{N}$ and encoding maps $\mathcal{F}_k$ and $\mathcal{G}_k$. Now, applying quantum Stein's Lemma to the quantum hypothesis testing problem with hypotheses $\mathsf{H}_0:\sigma_{W^k \widetilde{W}^k}$ and $\mathsf{H}_1: \widetilde{\sigma}_{W^k \widetilde{W}^k}$ yields
	\begin{align} \label{eq:Stein_general2}
	\lim_{\ell\to\infty} -\frac{1}{\ell}\log\beta(\ell ,\eps, \mathcal{F}_k,\mathcal{G}_k) = D\left( \sigma_{W^k \widetilde{W}^k}\left\| \widetilde{\sigma}_{W^k \widetilde{W}^k}\right. \right).
	\end{align}
	{{i.e.~}}for every $\delta>0$, there exists some $n_0\in\mathbb{N}$ such that for all $\ell \geq n_0$,
	\begin{align}
	D\left( \sigma_{W^k \widetilde{W}^k}\left\| \widetilde{\sigma}_{W^k \widetilde{W}^k}\right. \right)
	\leq -\frac{1}{\ell}\log\beta(\ell ,\eps, \mathcal{F}_k,\mathcal{G}_k) + \delta
	\end{align}
	
	Recalling the definition  ~\eqref{eq:theta} of $\theta_k$, dividing by $k$ on both sides of the above inequality, and
	taking the supremum over all $\mathcal{F}_k^{X^k\to W^k}$ and $\mathcal{G}_k^{Y^k\to \widetilde{W}^k}$ with constraints $\log |W^k|\leq k r_1$ and $\log |\widetilde{W}^k|\leq k r_2$, we have for all $\ell \geq n_0$,
	\begin{align}
	\theta_k(r_1,r_2) &\leq -\frac{1}{\ell k} \log \inf_{\substack{\mathcal{F}_k:\frac1k\log |W^k|\leq  r_1 \\ \mathcal{G}_k:\frac1k \log |\widetilde{W}^k|\leq r_2 } }  \beta(\ell,\eps, \mathcal{F}_k,\mathcal{G}_k)  + \frac{\delta}{k}.
	\end{align}
	Let $ n = \ell k $. 
	Considering $\beta$ as a function of $n$ (since $k$ is fixed) and allowing the encoding maps to be $\mathcal{F}_n^{X^n\to W^n}$ and $\mathcal{G}_n^{Y^n\to \widetilde{W}^n}$ (which satisfy the constraints $\mathcal{F}_n:\frac1n\log |W^n|\leq  r_1$ and $\mathcal{G}_n:\frac1n \log |\widetilde{W}^n|\leq r_2 $) instead of restricting the encodings to be on the $k$-blocklength systems, the right-hand side of the above inequality can be further upper bounded as follows: for any $n\ge n_0$
	\begin{align}
	\theta_k(r_1,r_2) &\leq 
		-\frac{1}{n} \log \inf_{\mathcal{F}_n,\, \mathcal{G}_n } \left\{ \beta(n,\eps, \mathcal{F}_n,\mathcal{G}_n): \log |W^n|\leq n r_1, \, \log |\widetilde{W}^n|\leq n r_2 \right\} + \frac{\delta}{k}\\
		&= 
	-\frac{1}{n} \log \beta_{r_1,r_2}(n,\eps) + \frac{\delta}{k}\,.
	\end{align} 
	Taking the limit inferior with respect to $n\to\infty$ and letting $\delta\to 0$ on both sides of the above inequality yields 
	\begin{align}
	\theta_k(r_1,r_2) &\leq \liminf_{n\to\infty} -\frac{1}{n} \log \beta_{r_1,r_2}(n,\eps).
	\end{align}
	Since this holds for arbitrary $k\in\mathbb{N}$, we obtain our first claim in  ~\eqref{eq:Stein_general0} by taking supremum over all $k\in\mathbb{N}$.
\item[b)]	Next, we show the optimality.
First, suppose there exists $n\in\mathbb{N}\cup\{\infty\}$, $\mathcal{F}_n^{X^n\to W^n}$, and $\mathcal{G}_n^{Y^n\to \widetilde{W}^n}$ such that $D\left( \sigma_{W^n \widetilde{W}^n}\left\| \widetilde{\sigma}_{W^n \widetilde{W}^n}\right. \right) = \infty$. Then $\theta(r_1, r_2) = \infty$ by the definition given in  ~\eqref{eq:theta}. In this case, the claim in  ~\eqref{eq:Stein_general1} holds trivially. In what follows, we consider $n\in\mathbb{N}$, $\mathcal{F}_n^{X^n\to W^n}$, and $\mathcal{G}_n^{Y^n\to \widetilde{W}^n}$ such that $D\left( \sigma_{W^n\widetilde{W}^n}\left\| \widetilde{\sigma}_{W^n \widetilde{W}^n}\right. \right)<\infty$.
	Let $0\leq T\leq \mathds{1}$ be an arbitrary test, which is to be specified later.
	Using the data-processing inequality of the quantum relative entropy $D(\cdot\|\cdot)$ with respect to the quantum channel
	\begin{align}
	\Lambda: \rho\mapsto \Tr\left[T\rho \right] \oplus \Tr\left[ (\mathds{1}-T) \rho \right], \quad \forall \, \rho\in \mathcal{D}(\cH_{W^n\widetilde{W}^n})\,,
	\end{align}
we get
	\begin{align}
	D\left( \sigma_{W^n \widetilde{W}^n}\left\| \widetilde{\sigma}_{W^n \widetilde{W}^n}\right. \right) &\geq D\left( \Lambda\left(\sigma_{W^n \widetilde{W}^n}\right) \left\| \Lambda\left(\widetilde{\sigma}_{W^n \widetilde{W}^n}\right)\right. \right)\\
	&= (1-\alpha) \log \frac{1-\alpha}{\beta} + \alpha\log \frac{\alpha}{1-\beta},
	\end{align}
	where $\alpha = \Tr\left[(\mathds{1}-T) \sigma_{W^n \widetilde{W}^n} \right]$ and $\beta= \Tr\left[ T \,\widetilde{\sigma}_{W^n \widetilde{W}^n}\right]$ are the type-I and the type-II errors, respectively.
	Now, we choose $\mathcal{F}_n^{X^n\to W^n}$, $\mathcal{G}_n^{Y^n\to \widetilde{W}^n}$, and $T$ such that  $\log |W^n|\leq n r_1$, $\log |\widetilde{W}^n|\leq n r_2$, $\alpha\leq \eps$, and $\beta = \beta_{r_1,r_2}(n,\eps)$. 
	From the definitions  ~\eqref{eq:theta} of $\theta$ and $\theta_n$, it follows that
	\begin{align}
	\theta(r_1,r_2) &\geq \theta_n(r_1,r_2) \geq \frac1n D\left( \sigma_{W^n \widetilde{W}^n}\left\| \widetilde{\sigma}_{W^n \widetilde{W}^n}\right. \right) \\
		&\geq \frac1n\left[ (1-\alpha) \log \frac{1-\alpha}{\beta_{r_1,r_2}(n,\eps)} + \alpha \log \frac{\alpha}{1-\beta_{r_1,r_2}(n,\eps) } \right]\\
		&\geq -\frac{1-\eps}{n} \log \beta_{r_1,r_2}(n,\eps) - \frac{\alpha}{n} \log \left[ 1- \beta_{r_1,r_2}(n,\eps) \right] - \frac{h(\alpha)}{n} \\
	&\geq -\frac{1-\eps}{n} \log \beta_{r_1,r_2}(n,\eps)- \frac{h(\alpha)}{n},
	\end{align}
	where we denote by $h(p) := -p\log p - (1-p)\log(1-p)$ the binary entropy function.
	Taking the limit inferior with respect to $n\to\infty$ and letting $\eps \to 0$ completes our proof in  ~\eqref{eq:Stein_general1}.
	\end{itemize}
\end{proof}

Theorem~\ref{theo:Stein_general} immediately yields the following result for testing against independence.
Let the following binary hypotheses.
\begin{align} \label{eq:independence}
\begin{dcases}
\mathsf{H}_0: \rho_{XY}^{\otimes n} \\
\mathsf{H}_1: \left(\rho_X\otimes \rho_Y\right)^{\otimes n} \\
\end{dcases}\,, \quad \forall\, n\in\mathbb{N}.
\end{align}

\begin{coro}[The Stein exponent for testing against independence] \label{coro:Stein_independence}
	Given the hypotheses of  ~\eqref{eq:independence}, it follows that for every $r_1, r_2\geq 0$, 
	\begin{align}
	\lim_{\eps\to 0} e(\eps|r_1,r_2) 
	= \sup_{n\in\mathbb{N}} \sup_{ \substack{ \mathcal{F}^{X^n\to W},\, \log|W|\leq nr_1 \\
			\mathcal{G}^{Y^n\to \widetilde{W}^n},\, \log|\widetilde{W}^n|\leq nr_2 }} \frac1n I\big( W^n;\widetilde{W}^n\big)_\sigma, 
	\end{align}
	where $\sigma$ is the state given in \eqref{eq:sigma0}, and 
	$I(A;B)_\rho := D(\rho_{AB}\|\rho_A\otimes \rho_B)$.
\end{coro}

\subsection{Single-letterization of the Stein exponent when there is no rate constraint for Bob} \label{sec:single-letter0}

In this section, we consider a scenario in which Charlie has access to the full information that Bob has, i.e.~$r_2 = \infty$. Moreover, we allow Alice and Charlie to have access to unlimited prior shared entanglement, given by a state $\Phi_{T_X T_C}$. Then, we prove that the Stein exponent in this scenario admits a single-letter formula.

\begin{theo}[A Single-Letter Formula] \label{theo:Stein_independence}
	Given the hypotheses of
	 ~\eqref{eq:independence}, it follows that for every $r> 0$
	\begin{align} \label{eq:Stein_independence0}
	\lim_{\eps\to 0} e(\eps| r,\infty) = \theta( r,\infty) = \sup_{ \substack{\mathcal{N}^{X\to U} \\ \frac12 I(U;YR)_{\omega} \leq r}  } I\left( U;Y\right)_\omega,
	\end{align}
	where 
	$\omega_{UYR} = (\mathcal{N}^{X\to U} \otimes \id_{YR}) \psi_{XYR} ,$
	and $\mathcal{N}^{X\to U}$ denotes a quantum channel from $X$ to $U$. 
	
\end{theo}

Theorem~\ref{theo:Stein_independence} immediately yields the following Corollary :
\begin{coro}
	Given the hypotheses of  ~\eqref{eq:independence}, it follows that for every $r\geq H(X)_\rho:= -\Tr\left[\rho_X\log \rho_X \right]$,
	\begin{align}
	\lim_{\eps\to 0} e(\eps| r,\infty) =  I\left( X;Y\right)_\rho. 
	\end{align}
\end{coro}
Before proving Theorem~\ref{theo:Stein_independence}, we need to introduce some quantities that relate our problem to the so-called \emph{quantum information bottleneck method} (see \cite{GS16,SCK+17,DHW18}). Consider the following quantity, defined in \cite{DHW18}, which is dual to the \emph{quantum information bottleneck function}: given $r\ge 0$
\begin{align}
I^\text{q}_Y(r) := \sup_{ \substack{\mathcal{N}^{X\to U} \\ I(U;YR)_{\omega} \leq r}  } I\left( U;Y\right)_\omega\,.
\end{align}

The assertion in Theorem~\ref{theo:Stein_independence} is equivalent to
\begin{align} \label{eq:Stein_independence00}
\lim_{\eps\to 0} e(\eps| r,\infty) \equiv \theta( r,\infty) = I_Y^\text{q}(2r), \quad \forall\, r>0\,.
\end{align}
Let us denote the Stinespring isometry of the encoding map $\mathcal{F}_n^{X^n\to W^n}$ by $\mathcal{U}_{\mathcal{F}_n}\equiv \mathcal{U}_{\mathcal{F}_n}^{X^n \to W^n E^n}$.
Then it can be easily verified that
\begin{align}
I(W^n;Y^n)_\sigma &= 2H(Y^n)_\sigma - I(Y^n;R^n E^n)_\sigma \\
&= 2nH(Y)_\rho - I(Y^n;R^n E^n)_\sigma\,, 
\end{align}
where the mutual information is evaluated for the state $\sigma_{E^nY^nR^n}$ which is a reduced state of
\begin{align} \label{eq:sigma}
\sigma_{W^nE^nY^nR^n} := \left(\mathcal{U}_{\mathcal{F}_n}^{X^n\to W^n E^n}\otimes \id_{Y^nR^n} \right) \psi_{XYR}^{\otimes n}\,,
\end{align}
with $\psi_{XYR}$ being a purification of $\rho_{XY}$, so that $H(W^nY^n)=H(E^nR^n)$ and $H(W^n)=H(E^nR^nY^n)$. Similarly, denoting the Stinespring isometry of the channel $\mathcal{N}^{X\to U}$ by $\mathcal{U}_{\mathcal{N}}\equiv \mathcal{U}_{\mathcal{N}}^{X \to UV}$, we have
\begin{align}
I(U;Y)_\omega = 2 H(Y)_\rho - I(Y;RV)_\omega
\end{align}
for the state
$\omega_{UVYR} := \left(\mathcal{U}_\mathcal{N}^{X\to UV}\otimes \id_Y\right) \psi_{XYR}.$
Then, by the definition of $\theta$ given in \eqref{eq:theta}, the claim in \eqref{eq:Stein_independence00} is equivalent to 
\begin{align} \label{eq:Stein_independence1}
\inf_{ \substack{ n\in\mathbb{N} \\ \mathcal{F}_n^{X^n\to W^n},\, \log|W^n|\leq n r} } \frac1n\, I\left( Y^n;R^n E^n\right)_{\sigma}
=
F_\text{q}(2r),
\end{align}
where for any $b\geq 0$,
\begin{align} \label{eq:Fq}
F_\text{q}(b) := \inf_{ \substack{\mathcal{N}^{X\to U} \\ I(U;YR)_{\omega} \leq b}  } I\left( Y;RV\right)_\omega = 2 H(Y)_\rho - I_Y^\text{q}(b).
\end{align}
We need the following property of $F_\text{q}(b)$ in order to establish our claim  ~\eqref{eq:Stein_independence1}.
\begin{lemm} \label{lemm:regularized}
	For every $n\in\mathbb{N}$ and $b\geq 0$, it holds that
	\begin{align} \label{eq:regularized0}
	F^n_\textnormal{q}(b) := \inf_{ \substack{\mathcal{F}_n^{X^n\to W^n} \\ \frac1n I(W;Y^n R^n)_{\sigma} \leq b}  } \frac1n I\left( Y^n; R^n E^n \right)_{\sigma} = F_\textnormal{q}(b),
	\end{align}
	where the state $\sigma_{W^nE^nY^nR^n}$ is given in  ~\eqref{eq:sigma}.
\end{lemm}
\begin{proof}
	We first show that the direction `$\leq$' in \eqref{eq:regularized0} follows from (i) the fact that the set of channels $\mathcal{F}_n^{X^n\to W^n}$ includes the composition of local operations, and (ii) by additivity of the mutual information under tensor product.

More precisely, let $W^n = U_1\ldots U_n$ and $E^n = V_1\ldots V_n$ and let 
$$\mathcal{U}_{\mathcal{F}_n}^{X^n \to W^nE^n}= \bigotimes_{i=1}^n \mathcal{U}^{X_i \to U_iV_i},$$  
where $\mathcal{U}^{X_i \to U_iV_i}\equiv \mathcal{U}_{\mathcal{N}}^{X_i \to U_iV_i}$ denotes the isometry of $\mathcal{N}^{X_i \to U_i}$. Then defining the state $\omega_{UVYR} := \left(\mathcal{U}_\mathcal{N}^{X\to UV}\otimes \id_Y\right) \psi_{XYR}$, we have that
	\begin{align}
	\sigma_{W^nY^nR^n E^n} &=  \left(\mathcal{U}_{\mathcal{F}_n}^{X^n\to W^nE^n}\otimes \id_{Y^n} \right) \psi_{XYR}^{\otimes n} \\
	&=  \left(\bigotimes_{i=1}^n \mathcal{U}_{\mathcal{N}}^{X_i\to U_iV_i}\otimes \id_{Y^n} \right) \psi_{XYR}^{\otimes n} \\
	&= (\omega_{UVYR})^{\otimes n}\,.
	\end{align}
	Additivity of the mutual information under tensor products then yields the following
	\begin{align}
	\begin{split} \label{eq:regularized1}
	\frac1n I(W^n;Y^n R^n)_{\sigma} &= \frac1n I(U_1\ldots U_n; Y^n R^n)_{\omega^{\otimes n}} = I(U;YR)_\omega; \\
	\frac1n I\left( Y^n; R^n E^n\right)_{\sigma} &= \frac1n I\left( Y^n; R^n V_1\ldots V_n \right)_{\omega^{\otimes n}} = I(Y;RV)_\omega.
	\end{split}
	\end{align}
	This implies that
	\begin{align}
	F^n_\textnormal{q}(b)
	&\leq \inf_{  \substack{ \mathcal{N}^{X \to U}: \mathcal{F}_n^{X^n\to W^n }= (\mathcal{N}^{X \to U})^{\otimes n} \\ \frac1n I(W^n;Y^n R^n)_{\sigma} \leq b}  } \frac1n I\left( Y^n; R^n E^n\right)_{\sigma} \\
	&=  \inf_{ \substack{\mathcal{N}^{X\to U} \\ I(U;YR)_{\omega} \leq b}  } I\left( Y;RV\right)_\omega \\
	&= F_\text{q}(b)\,,
	\end{align}
	which proves the direction `$\leq $'.

	To show the other direction, i.e.~`$\geq$', we will employ an additivity property proved in Ref.~\cite{DHW18} which is described below.
	Define the set $\mathcal{T}(\psi)$ for any pure state $\psi\equiv \psi_{XYR}$:
	\begin{align}\label{eq:set_T}
	\mathcal{T}(\psi) := \left\{(Q_X, Q_Y): \exists\, \mathcal{U}^{X\to U V} \text{ isometry s.t. } 
	Q_X \geq I(U;YR)_\omega, \, 
	Q_Y \geq I(Y;RV)_\omega, \, 
	\right\},
	\end{align}
	where $	\omega_{UVYR} := (\mathcal{U}^{X\to U V} \otimes \id_{YR}) \psi_{XYR}$.
	Recalling the definitions of $F_\text{q}(b)$ and $F_\text{q}^n(b)$ given in  ~\eqref{eq:Fq} and \eqref{eq:regularized0}, one can rewrite them as
	\begin{align}
	F_\text{q}(b) &= \inf_{Q_Y:\,(b,Q_Y) \in \mathcal{T}(\psi) } Q_Y; \label{eq:regularized2}\\
	F_\text{q}^n(b) &= \inf_{ Q_{Y^n}:\,(nb, Q_{Y^n}) \in \mathcal{T}(\psi^{\otimes n}) } \frac1n Q_{Y^n}. \label{eq:regularized3}
	\end{align}
	
	On the other hand, Theorem~2 in \cite{DHW18} states that for any pure state $\psi_{XYR}$ and $n\in\mathbb{N}$,
	\begin{align}
	\mathcal{T}(\psi^{\otimes n}) = \underbrace{\mathcal{T}(\psi) + \cdots + \mathcal{T}(\psi)}_{n \text{ terms}}\,,
	\end{align}
	where the `$+$' refers to the element-wise sum of sets\footnote{We note that the set the set $\mathcal{T}(\psi)$ in \eqref{eq:set_T} is slightly different from the one introduced in Theorem~2 of \cite{DHW18}. Namely, we do not have the additional factor $2$ in front of $Q_X$ and $Q_Y$. It can be verified that the additivity of the set holds for all positive factors. Hence, we choose factor $1$ for the purpose of our proof.}.
	This means that any element $(nb, Q_{Y^n}) \in \mathcal{T}(\psi^{\otimes n}) $ satisfies
	\begin{align} \label{eq:regularized4}
	\begin{dcases}
	nb = \sum_{i=1}^n b_i \\
	Q_{Y^n} = \sum_{i=1}^n Q_{Y_i}\\
	\end{dcases}, \quad \forall\, (b_i, Q_{Y_i}) \in \mathcal{T}(\psi), \; i\in[n].
	\end{align}
	Using this fact on  ~\eqref{eq:regularized3}, we have
	\begin{align}
	F_\text{q}^n(b) &= \inf_{ \substack{(b_i, Q_{Y_i}) \in \mathcal{T}(\psi), \; \forall i\in[n] \\ nb = \sum_{i=1}^n b_i } } \frac1n \sum_{i=1}^n Q_{Y_i} \\
	&= \inf_{(b_i)_{i=1}^n:\; nb = \sum_{i=1}^n b_i } \inf_{Q_{Y_i}:\, (b_i, Q_{Y_i}) \in \mathcal{T}(\psi) } \,\frac1n \sum_{i=1}^n Q_{Y_i} \\
	&= \inf_{(b_i)_{i=1}^n:\; nb = \sum_{i=1}^n b_i } \frac1n\sum_{i=1}^n F_\text{q}(b_i) \\
	&\geq \inf_{(b_i)_{i=1}^n:\; nb = \sum_{i=1}^n b_i } \,\sum_{i=1}^n F_\text{q}\left( \frac1n \sum_{i=1}^n b_i\right) \\
	&= F_\text{q}(b),
	\end{align}
	where the third equality follows from the definition of $F_\text{q}(b)$ given in  ~\eqref{eq:regularized2}, and the sole inequality follows from the convexity of $b\mapsto F_\text{q}(b)$, which is proved in Theorem~1 of \cite{DHW18}. This  completes our claim.
\end{proof}

Now, we are ready to prove Theorem~\ref{theo:Stein_independence}.
\begin{proof}[Proof of Theorem~\ref{theo:Stein_independence}]
	
	(Achievability): 
	To prove the direction `$\leq$' in  ~\eqref{eq:Stein_independence1},
	it suffices to show that for every channel $\mathcal{N}^{X\to U}$ satisfying $\frac12 I(U;YR)_\omega \leq r$ and every $\delta >0$, there exist an $n\in\mathbb{N}$ and an encoding map $\mathcal{F}_n^{X^n\to W^n}$ such that the following hold:
	\begin{align}
	\frac1n \log |W^n| &\leq r + \delta, \label{eq:Stein_independence_ach1} \\
	\frac1n I(Y^n; R^n E^n)_\sigma &\leq I(Y;RV)_\omega + \delta.  \label{eq:Stein_independence_ach2}
	\end{align}
	
	To prove this, we employ a quantum state splitting protocol (or equivalently the Quantum Reverse Shannon Theorem (QRST)~\cite{6757002}) as follows. Note that we make explicit use of the entangled state shared between Alice and Charlie, which we denote as $\Phi_{T_XT_C}$, with $T_X$ being with Alice and $T_C$ being with Charlie.
	Alice generates the system $W^n$ by acting locally on $T_X X^n$, and then sends $W^n$ to Charlie. Hence, the encoding map is given by $\mathcal{F}_n^{T_X X^n \to W^n}$, and we denote its Stinespring isometry by $\mathcal{U}_{\mathcal{F}_n}^{T_X X^n \to W^n E^n}$.
	Charlie first applies an isometry on the received system $W^n$ and his share of entanglement $T_C$, which we denote by $\mathcal{V}^{W^n T_C \to U^n}$. 
	The aim is to have the final state
	\begin{align}
	\tau_{U^nY^nR^n E^n} := \left( \mathcal{V}^{W^n T_C \to U^n} \circ \mathcal{U}_{\mathcal{F}_n}^{T_X X^n \to W^n E^n} \otimes \id_{Y^n R^n} \right) \left( \psi_{XYR}^{\otimes n}\otimes \Phi_{T_X T_C} \right) \approx 
	\omega_{UVYR}^{\otimes n}
	\end{align}
	close to the state $\omega_{UVYR}^{\otimes n}$
	According to the Quantum Reverse Shannon Theorem~\cite{6757002}, to achieve this aim, Alice needs to send qubits at a rate $\frac1n \log |W^n| = \frac12 I(U;YR)_\omega + \delta$ to Charlie for sufficiently large $n\in\mathbb{N}$.
	By the additivity of the mutual information, we obtain
	\begin{align*}
	\frac1n I(Y^n; R^n E^n)_{\omega^{\otimes n}} = I(Y;RV)_\omega.
	\end{align*}
	 Since the mutual information is continuous with respect to the underlying state, {{there must exist}} sufficiently large $n\in\mathbb{N}$ such that
	 \begin{align}
	\frac1n I(Y^n; R^n E^n)_\tau 
	 &\leq \frac1n I(Y^n; R^n E^n)_{\omega^{\otimes n}} + \delta \\
	 &= I(Y;RV)_\omega  + \delta.
	\end{align}
	Noting that $\frac1nI(Y^n; R^n E^n)_\sigma = \frac1n I(Y^n; R^n E^n)_\tau$, 
	inequality~\eqref{eq:Stein_independence_ach2} is thus proved.
	
	On the other hand, by the assumption of $\frac12 I(U;YR)_\omega \leq r$ for the channel $ \mathcal{N}^{X\to U}$,  ~\eqref{eq:Stein_independence_ach1} is also satisfied, i.e.~
	\begin{align}
	\frac1n \log |W^n| = \frac12 I(U;YR)_\omega + \delta \leq r + \delta.
	\end{align}
	Hence, we have proved the `$\leq$' of \eqref{eq:Stein_independence1}.
	
	(Optimality/Converse): We move on to the other direction `$\geq$' on  ~\eqref{eq:Stein_independence1}.
	Let $n\in\mathbb{N}$ and $\mathcal{F}^{X^n \to W^n}$ be arbitrary. Denote by $ {Q}_X = \frac1n \log |W^n|$ and $ {Q}_Y = \frac1n I(Y^n;R^nE^n)_\sigma$.
	We then have to show that
	\begin{align} \label{eq:Stein_independence_opt3}
	{Q}_Y \geq \inf_{ \substack{\mathcal{N}^{X\to U} \\ \frac12 I(U;YR)_{\sigma} \leq  {Q}_X}  } I\left( Y;RV\right)_\sigma.
	\end{align} 
	
	To that end, we first give a lower bound on $ {Q}_X$:	
	\begin{align}
	n {Q}_X 
	&\geq H(W^n)_\sigma \\
	&= \frac12 I(W^n;E^n)_\sigma + \frac12 I(W^n; T_C Y^n R^n)_\sigma \\
	&\geq \frac12 I(W^n; T_C Y^n R^n)_\sigma\\
	&= \frac12 I(W^n;Y^n R^n | T_C) + \frac12 I(W^n; T_C)_\sigma \\
	&\geq \frac12 I(W^n;Y^n R^n | T_C)_\sigma \\
	&=  \frac12 I(W^n T_C;Y^n R^n )_\sigma - \frac12 I(T_C; Y^n R^n)_\sigma \\
	&\geq  \frac12 I(W^n;Y^n R^n )_\sigma. \label{eq:Stein_independence_opt1}
	\end{align}
	where the fist inequality follows from $H(W^n)\leq \log |W^n|$, the second line follows from the fact that for any pure state of a tripartite system $ABE$,
	$H(A) = \frac12 I(A;B) + \frac12 I(A;E).
	$
	The third and the fifth line are due to the positivity of mutual information. 
	In the fourth and the sixth lines, we use the chain rule
	$I(A;BC) = I(A;B) + I(A;C|B)$;
	the last line holds because $T_C$ is uncorrelated with $Y^n R^n$, and we use the data-processing inequality with respect to partial trace.
	
	Next, we prove the desired lower bound given in  ~\eqref{eq:Stein_independence_opt3}: 
	\begin{align}
	{Q}_Y &= \frac1n I(Y^n; R^n E^n)_\sigma \\
	&\geq F_\text{q}^n\left( \frac1n I(W^n;Y^n R^n)_\sigma \right) \\
	&\geq F_\text{q}^n\left( 2 Q_X \right) \\
	&= F_\text{q}(2 Q_X)  \\
	&=  \inf_{ \substack{\mathcal{N}^{X\to U} \\ \frac12 I(W;YR)_{\sigma} \leq Q_X}  } I\left( U;Y\right)_\sigma,
	\end{align}
	where the first inequality follows from the definition of $F_\text{q}^n(b)$ given in  ~\eqref{eq:regularized0} in the first inequality. The second inequality follows from inequality \eqref{eq:Stein_independence_opt1} and the fact that $b\mapsto F_\text{q}^n(b)$ is monotonically decreasing. In the fourth line, we apply Lemma~\ref{lemm:regularized}.
	Hence, we prove our assertion of Theorem~\ref{theo:Stein_independence}.
\end{proof}

\section{Second order converses in quantum network theory via reverse hypercontractivity} \label{sec:sc}

In classical network information theory, the so-called {\em{Blowing-Up Lemma (BUL) method}} has proved to be very useful in yielding strong converse bounds for various tasks for which traditional methods such as single-shot and type class analysis are known to fail (see \cite{AGK76,Tan14,CK11,raginsky2013concentration}). The BUL method relies on the concentration of measure phenomenon (see \cite{ledoux2001concentration,BLM13,marton1966simple,raginsky2013concentration}), which implies that the tail probability of a sequence being far (in Hamming distance) from a correctly decoded sequence is sub-Gaussian. Using this fact, the decoding sets can be slightly \textit{blown up} in such a way that the probability of decoding error can be made arbitrarily small for sufficiently large blocklengths. The connection of those sets with the achievable rate of the task under study, generically expressed in terms of an entropic quantity, is then made by means of data-processing inequality. Despite being widely applicable, this method suffers two major drawbacks: first, it only yields suboptimal second-order terms (typically of order $\mathcal{O}(\sqrt{n}\,\log^{\frac{3}{2}}(n))$). Secondly, the argument is restricted to finite alphabets. 

Recently, another approach based on functional inequalities related to the phenomenon of concentration of measure has been developed in \cite{LHV18,liu2017beyond}. There, instead of \textit{blowing-up} decoding sets, the idea is to work with indicator functions of those sets and \textit{smooth} them out by perturbing them through the action of a Markov semigroup. As opposed to the blowing-up lemma, here one uses a variational formulation for the entropic expression governing the rate of the task under study, which involves an optimization over a class of (typically positive) functions. The connection to the decoding sets is then made by choosing the optimizing function to be the smoothed version of the indicator function over that set. The advantage of this method is that it avoids the use of the data-processing inequality, which is responsible for the weaker second-order term achieved by the blowing-up method. Not only does this method provide the right second-order term (typically $\mathcal{O}(\sqrt{n})$) in the strong converse bound, but the finer control provided by this Markovian approach can be extended to general alphabets, Gaussian channels and channels with memory. Moreover, and as we will see now, the functional analytical nature of the smoothing-out method is more easily generalizable to the quantum setting where the set of allowed tests (or POVMs) is strictly larger than the one of those that are diagonalizable in the basis of decoded sequences.

In this section, we derive second-order, finite blocklength strong converse bounds for the task of bipartite hypothesis testing under communication constraints which was introduced in Section \ref{sec:Stein}, as well as for the task of quantum source coding with classical side information at the decoder. As mentioned before, our techniques rely on the tensorization property of the reverse hypercontractivity for the quantum generalized depolarizing semigroup. For sake of clarity, we provide a brief introduction to the techniques that are being used in the next subsection (see \cite{trunck} and the references therein for more details).

\subsection{Quantum Markov semigroups and reverse hypercontractivity}

In this section, we employ a powerful analytical tool, namely, the so-called {\em{quantum reverse hypercontractivity}} of a certain quantum Markov semigroup (QMS) and its tensorization property. Let us introduce these concepts and the relevant results in brief. For more details see e.g.~\cite{BDR18} and references therein. 

We recall that a \textit{quantum Markov semigroup} (QMS) $(\Phi_t)_{t\ge 0}$ on the algebra $\mathcal{B}(\cH)$ of linear operators on the Hilbert space $\cH$ is a family of quantum channels (in the Heisenberg picture), that consists of completely positive, unital maps $\Phi_t:\mathcal{B}(\cH)\to\mathcal{B}(\cH)$ such that
\begin{itemize}
	\item[(i)] For any $t,s\ge 0$, $\Phi_{t+s}=\Phi_t\circ \Phi_s$;
	\item[(i)] $\Phi_0={\rm{id}}$, where ${\rm{id}}$ denotes the identity map on $\mathcal{B}(\cH)$ ;
	\item[(ii)] The family $(\Phi_t)_{t\ge 0}$ is strongly continuous at $t=0$.
\end{itemize}
When $\cH$ is finite dimensional, there exists a linear map $\mathcal{L}:\mathcal{B}(\cH)\to\mathcal{B}(\cH)$ such that $$\mathcal{L}=\frac{d\Phi_t}{dt}\,\,\,\,\,\Leftrightarrow\,\,\,\,\, \Phi_t:=\e^{-t\mathcal{L}}\,.$$
We further assume that the QMS is \textit{primitive}, that is, there exists a unique full-rank state $\sigma$ which is invariant under the evolution: $$\forall t\ge 0,\,\,\,\Phi_t(\sigma)=\sigma\,.$$

The QMS that we consider is the so-called {\em{generalized quantum depolarizing semigroup}}. In the Heisenberg picture, for any state $\sigma >0$ on a Hilbert space $\cH$, the generalized quantum depolarizing semigroup with invariant state $\sigma$ is defined by a one-parameter family of linear completely positive (CP) unital maps $\left(\Phi_t\right)_{t \geq 0}$, such that for any $X \in \cB(\cH)$,
\begin{align}\label{eq-gqds}
\Phi_t(X) = \e^{-t}X + (1-\e^{-t}) \tr(\sigma X) \,\mathbb{I},
\end{align}
In the Schr\"odinger picture, the corresponding QMS is given by the family of CPTP maps $\left(\Phi^\star_t\right)_{t \geq 0}$, such that
$$ \tr(Y \Phi_t (X)) = \tr( \Phi^\star_t(Y) X), \quad \forall, \,\, X,Y \in \cB(\cH).$$
The action of $\Phi^\star_t$ on any state $\rho \in \cD(\cH)$ is that of a generalized depolarizing channel, which keeps the state unchanged
with probability $e^{-t}$, and replaces it by the state $\sigma$ with probability $(1 - e^{-t})$:
$$ \Phi^\star_t(\rho) = \e^{-t}\rho + (1-\e^{-t}) \sigma\,.$$
Note that $ \Phi^\star_t(\sigma) = \sigma$ for all $t\geq 0$, and that $\sigma$ is the unique invariant state of the evolution.
\medskip

To state the property of quantum reverse hypercontractivity, we define, for any $X\in\mathcal{B(H)}$, the non-commutative weighted $L_p$ norm {with respect to} the state $\sigma \in \mathcal{D}(\cH)$, for any $p \in \mathbb{R}\backslash\{0\}$\footnote{ {For $p<1$, these are pseudo-norms, since they do not satisfy the triangle inequality. For $p<0$, they are only defined for $X>0$ and for a non-full rank state by taking them equal to $\big( \Tr\big[ \big| \sigma^{-\frac{1}{2{p}}} X^{-1} \sigma^{-\frac{1}{2 {p}}} \big|^{{-p}} \big]\big)^{1/p}$}.}:
\begin{align} \label{eq:wLp}
\left\| X \right\|_{p, \sigma} := \left( \Tr\left[ \left| \sigma^{\frac{1}{2{p}}} X \sigma^{\frac{1}{2 {p}}} \right|^{{p}} \right] \right)^{\frac{1}{{p}}}.
\end{align} 
\smallskip
A QMS $\left(\Phi_t\right)_{t \geq 0}$ is said to be reverse $p$-contractive for $p <1$, if
\begin{align}\label{contract}
|| \Phi_t(X)||_{p,\sigma} \geq ||X||_{p,\sigma}, \quad \forall \, X >0.
\end{align}
The generalized quantum depolarizing semigroup can be shown to satisfy a stronger inequality: $\forall$ $ p<q <1$,
\begin{align}\label{QRHC}
|| \Phi_t(X)||_{p,\sigma} &\geq ||X||_{q,\sigma}, \quad \forall \, X >0,
\end{align}
for
\begin{align}\label{tcond}
t &\geq \frac{1}{4 \alpha_1 (\cL)} \log \left(\frac{p-1}{q-1}\right),
\end{align}
where $\alpha_1 (\cL)>0$ is the so called {\em{modified logarithmic Sobolev constant}}, and $\cL$ denotes
the generator of the generalized quantum depolarizing semigroup, which is defined through the relation $\Phi_t(X) = e^{-t\cL}(X)$ and is given by
$$\cL(X) = X - \tr(\sigma X) \mathbb{I}.$$
The inequality \eqref{QRHC} is indeed stronger than \eqref{contract} since the map $p \mapsto ||X||_{p,\sigma}$ is non-decreasing.
\medskip

In the context of this paper, instead of the generalized quantum depolarizing semigroup defined through \eqref{eq-gqds}, we need to consider the QMS $\left(\Phi_{t, x^n}\right)_{t \geq 0}$,
with $\Phi_{t, x^n}$ being a CP unital map acting on $\cB(\cH^{\otimes n})$, and being labelled by sequences
$x^n \equiv (x_1,x_2, \ldots, x_n) \in \cX^n$, where $\cX$ is a finite set. For any $x \in \cX$, let $\rho^x \in \cD (\cH)$. Further, let
\begin{align}
\rho^{x^n} &:= \rho^{x_1} \otimes \cdots \otimes \rho^{x_n} \,\in \cD (\cH^{\otimes n}).
\end{align}
Then,
\begin{align}
\Phi_{t,x^n} &:= \Phi_{t,x_1} \otimes \cdots \otimes \Phi_{t,x_n},
\end{align}
where $(\Phi_{t,x})_{t\ge 0}$ is a generalized quantum depolarizing semigroup with invariant state $\rho^x$.
We denote by $\cK_{x^n}=\sum_{i=1}^n\widehat{\cL}_{x_i}$ the generator of $(\Phi_{t,x^n})_{t\ge 0}$ where $\widehat{\cL}_{x_i}={\rm{id}}^{\otimes i-1} \otimes  {\cL}_{x_i} \otimes {\rm{id}}^{\otimes n-i}$, with ${\cL}_{x_i}$ being the generator of the generalized quantum depolarizing semigroup $(\Phi_{t,x_i})$. The following quantum reverse hypercontractivity of the above tensor product of generalized quantum depolarizing semigroup was established in~\cite{BDR18} ( { {See~\cite{mossel2013reverse} for its classical counterpart, as well as \cite{CKMT15} for its extension to doubly stochastic QMS}}):
\begin{theo}[Quantum reverse hypercontractivity for tensor products of depolarizing semigroups {\cite[Corollary 17, Theorem 19]{BDR18}}] \label{lemm:RHC}
	For the QMS $\left(\Phi_{t, x^n}\right)_{t \geq 0}$ introduced above, for any $\mathsf{p}\leq \mathsf{q}< 1$ and for any $t$ satisfying $t\geq \log \frac{\mathsf{p}-1}{\mathsf{q}-1}$, the following inequality holds:
	\begin{align}
	\left\| \Phi_{t,x^n}(G_n) \right\|_{\mathsf{p},\rho^{x^n}}
	\geq \left\| G_n \right\|_{\mathsf{q},\rho^{x^n}}, \quad \forall\, G_n>0.
	\end{align}
	In other words, $\alpha_1(\cK_{x^n})\ge \frac{1}{4}$.
\end{theo}
The following two inequalities play key roles in our proofs.
\begin{lemm}
	[Araki-Lieb-Thirring inequality {\cite{Ari76, LT76}}] \label{lemm:ALT}
	For any $A,B\in \mathcal{P}(\mathcal{H})$, and $r\in[0,1]$,
	\begin{align}
	\Tr\left[ B^{\frac{r}{2}} A^r B^{\frac{r}{2}} \right] \leq \Tr\left[ \left( B^\frac12 A B^\frac12 \right)^r\, \right].
	\end{align}
\end{lemm}

\begin{lemm}
	[Reverse H\"older's inequality {\cite[Lemma 1]{BDR18}}] \label{lemm:RHI}
	Let $A\geq 0$ and $B>0$. Then, for any $p<1$ with H\"older conjugate $\hat{p} := 1/(1-1/p)$, we have
	\begin{align}
	\langle A, B \rangle_\sigma \geq \left\| A \right\|_{p,\sigma} \left\| B \right\|_{\hat{p},\sigma}.
	\end{align}
\end{lemm}

We are now ready to state and prove the first main result of this section, namely a second-order finite blocklength strong converse bound for the task of bipartite quantum hypothesis testing under communication constraints, as defined in Section \ref{sec:Stein}.

\subsection{Distributed quantum hypothesis testing when there is no rate constraint for Bob}\label{strongconverseQHT}

Consider the binary hypotheses given in  ~\eqref{eq:independence}. 
In this section, we further make the following assumption: $$\rho_{XY} = \sum_{x} Q_X(x)|x\rangle \langle x|\otimes \rho_Y^x $$ is a classical-quantum (c-q) state, where $x$ takes values in a finite set $\mathcal{X}$, and $Q_X$ denotes a probability distribution on $\cX$. For any $x \in \cX$, the state $\rho_Y^x$ can 
be viewed as the output of a c-q channel $\Lambda\equiv\Lambda^{X \to Y}$:
\begin{align}\label{cqxy}
\Lambda(x) = \rho_Y^x \in \mathcal{D}(\mathcal{H}_Y), \quad \forall x \in \mathcal{X}.
\end{align}

In this section we use $X$ and $U$ to denote both random variables (taking values in finite sets $\cX$ and $\cU$, respectively), as well as quantum systems whose associated Hilbert spaces, $\cH_X$ and $\cH_U$, have complete orthonormal bases $\{|x\rangle\}$ and $\{|u\rangle\}$ labelled by the values taken by these random variables. We also consider $X'$ to be a copy of the random variable $X$, and consider the following extension of the state $\rho_{XY}$:
\begin{align}
\rho_{XXY} := \sum_{x} Q_X(x)|x\rangle \langle x|\otimes|x\rangle \langle x|\otimes \rho_Y^x  \in \cD(\cH_X \otimes \cH_{X} \otimes \cH_Y). 
\end{align}
For a stochastic map (i.e.~classical channel) $\mathcal{N}^{X\to U}$, we define the state
\begin{align}\label{omega}
\omega_{UXY} := \left(\mathcal{N}^{X\to U} \otimes {\rm{id}}_{X} \otimes {\rm{id}}_{Y}\right) \rho_{XXY}.
\end{align}

Theorem~\ref{theo:sc_cq} gives the second-order strong converse for the Stein's exponent in this case. Note that without loss of generality, we can assume $Q_X$ to have full support\footnote{This is because the statement of the theorem holds trivially if $Q_X$ does not have full support.}.

\begin{theo}
	[Strong converse bound for the Stein exponent] \label{theo:sc_cq}
	For any $r>0$ and $\eps\in(0,1)$, 
	\begin{align}
	-\frac1n \log &\,\beta_{r,\infty}(n,\eps)\nonumber \\
	&\leq \sup_{ \substack{\mathcal{N}^{X\to U} \\  I(X';U)_\omega \leq r}  } I\left( U;Y\right)_\omega + 
	\left( 2 \log ( \gamma \eta ) \sqrt{ 3 \eta \log \frac{4|\mathcal{X}|}{1-\eps} } + 2 \sqrt{ 2 \gamma \log \frac{4}{1-\eps} } \right) \frac{1}{\sqrt{n}} +\frac{ 2}{n}\, \log \frac{4}{1-\eps},\label{eq-theo}
	\end{align}
		for $n> 3\eta \log \frac{4|\mathcal{X}|}{1-\eps}$, 
	where 
	$\eta := \max_{x \in \cX} \frac{1}{Q_X(x)}$
	and $\gamma:= \max_x \left\| \rho_Y^x \rho_Y^{-1} \right\|_\infty$. In the above, $\omega \equiv \omega_{UX'Y}$ is the state defined in (\ref{omega}).
\end{theo}

\noindent
\begin{remark}Note that in the classical setting, Ahlswede and Csisz\'ar proved the strong converse property by means of the blowing-up lemma \cite{AC86}. A better bound was recently found by \cite{LHV18} using the smoothing-out method. As mentioned at the beginning of this section, we follow the latter method.
	\end{remark}

\begin{proof}
	
	Using the Lagrange multiplier method, the first-order term on the right-hand side of (\ref{eq-theo}) can be re-written as
	\begin{align}
	\theta(r, \infty) 
	& := \inf_{c>0} \sup_{ \mathcal{N}^{X \to U} } \left\{ I(U;Y)_\omega - \frac{1}{c}\left(I(U;X')_\omega - r \right)
	\right\} 
\label{first}
\end{align}
Note that $I(U;Y)_\omega \leq I(U;X)_\omega$, by the data-processing inequality for the mutual information under the c-q channel $\Lambda^{X \to Y}$ defined in \eqref{cqxy}. This implies that $\forall$ $c \leq 1$,
$$I(U;Y)_\omega - \frac{1}{c}\left(I(U;X)_\omega \right)\leq 0.$$
The supremum of this quantity is equal to zero and is attained by a stochastic channel $\mathcal{N}^{X \to U}$ for which $U$ is independent of $X$ and $Y$. This allows us to restrict the infimum in (\ref{first}) to $c \geq1$: 
\begin{align}
\theta(r) 	&= \inf_{c\geq1} \sup_{ \mathcal{N}^{X \to U} } \left\{ I(U;Y)_\omega - \frac{1}{c}\left(I(U;X)_\omega - r \right)
	\right\}, \label{eq:first-order_c-q}
	\end{align}
	
	Hence, to prove the theorem, it suffices to show that, for every (possibly random) classical encoder $\mathcal{F}_n^{ {X}^n\to {W}^n}$, i.e.~a stochastic map defined through the conditional probabilities $\{P_{W^n|X^n}(w|x^n)\}$, with the random variable ${W}^n$ taking values $w \in \cW^n$, 
such that $\frac1n\log |\mathcal{W}^n| \leq r$, and any c-q test $T_{W^n Y^n} := \sum_{w\in\mathcal{W}} |w\rangle\langle w| \otimes T_{Y^n}^w$ for $0\leq T_{Y^n}^w \leq \mathds{1}_{Y^n}$ satisfying 
	\begin{align} \label{eq:prarameter_c-q}
	\begin{dcases}
	\mathrm{Pr}_T \{\mathsf{H}_1|\mathsf{H}_0\}:=\tr\big(\sigma_{W^nY^n}\,(\mathds{1}-T_{W^nY^n})  \big) \leq \eps \\
	\mathrm{Pr}_T \{\mathsf{H}_0|\mathsf{H}_1\} := \tr\big(\sigma_{W^n}\otimes\rho_{Y^n}\,T_{W^nY^n} \big) = \beta
	\end{dcases}\,\,,
	\end{align}
	(here, we have put a superscript to highlight the dependence on $T$),
	the following holds for all $c\geq1$,
	\begin{align} \label{eq:sc_c-q}
		-\frac1n \log \beta &\leq 
	\sup_{ \mathcal{N}^{X \to U} } \left\{ I(U;Y)_\sigma - \frac{1}{c} I(U;X)_\sigma +\frac{1}{c\,n}\,\log |\mathcal{W}^n| \right\} + \frac{K_\eps}{\sqrt{n}}+\frac{2}{n}\,\log\frac{4}{1-\eps}.
	\end{align}
	for $K_\eps :=  2 \log ( \gamma \eta ) \sqrt{ 3 \eta \log \frac{4|\mathcal{X}|}{1-\eps} } + 2 \sqrt{ 2 \gamma \log \frac{4}{1-\eps} }$.
	\medskip
We divide the proof of (\ref{eq:sc_c-q}) into three steps, given below.	
\smallskip

	\noindent\textbf{Step 1.}
	For an arbitrary $\eps' \in (0,1-\eps)$, one can use an expurgation argument to construct a new test $\tilde{T}_{W^n Y^n}:= \sum_w |w\rangle\langle w| \otimes \tilde{T}_{Y^n}^w $ such that for all $w\in \mathcal{W}^n$ 
	\begin{align} \label{eq:new_parameter_c-q}
	\begin{dcases} 
	\mathrm{Pr}_{\tilde{T}}\{\mathsf{H}_1|\mathsf{H}_0\} \leq \eps + \eps' \\
	\mathrm{Pr}_{\tilde{T}}\{\mathsf{H}_0|\mathsf{H}_1, w\}:=    \tr\big(\rho_{Y^n}\tilde{T}^{w}_{Y^n}) \big)\leq  \frac{\beta}{\eps'}\,.
	\end{dcases}
	\end{align}
Here $\mathrm{Pr}_{\tilde{T}}\{\mathsf{H}_0|\mathsf{H}_1, w\}$ denotes the type-II error under the condition that the decoder (Charlie) receives the classical index $w$, i.e.~the random variable $W^n$ takes the value $w$.

	The proof of the expurgation method is similar to the one used in the classical case~\cite{LHV18} and is deferred to Appendix~\ref{sec:expurgation}.
	\\~\\~
	\noindent\textbf{Step 2.}
	From  ~\eqref{eq:new_parameter_c-q}, we have
	\begin{align}
	1-\eps - \eps' 
	&\leq \mathrm{Pr}_{\tilde{T}}\{\mathsf{H}_0|\mathsf{H}_0\} \\
	&= \sum_{x^n} Q_{X}^{\otimes n}(x^n) \sum_{ w\in \mathcal{W}^n} P_{W^n|X^n}(w|x^n)  \Tr\left[ \rho_{Y^n}^{x^n} {\tilde{T}}_{Y^n}^w \right] \\
	&\leq  \sum_{x^n} \mu_n(x^n) \sum_{ w\in\mathcal{W}^n} P_{W^n|X^n}(w|x^n)  \Tr\left[ \rho_{Y^n}^{x^n} {\tilde{T}}_{Y^n}^w \right]  + \delta,
	\end{align}
	where we introduce a new measure 
\label{mun}
\begin{align}
\mu_n &:=   Q_{X}^{\otimes n}|_{\mathcal{C}_n} := Q_{X}^{\otimes n} \mathrm{1}_\{x^n \in \mathcal{C}_n\},
\end{align} 
for some set $\mathcal{C}_n \subseteq \mathcal{X}^n$ satisfying 
	$Q_{X}^{\otimes n} [\mathcal{C}_n] := \sum_{x^n \in \cX^n} Q_{X}^{\otimes n}(x^n) \geq 1-\delta$, for some $\delta >0$.
	The set $\mathcal{C}_n$ and the constant $\delta$ is specified later at Step 3.
	
	For any $t>0$ and $c\geq 1$, we have
	\begin{align}
	&(1-\eps-\eps' - \delta)^{ c(1+\frac1t)}\nonumber \\
	&\leq \left(  \sum_{x^n} \mu_n(x^n) \sum_{ w\in \mathcal{W}^n } P_{W^n|X^n}(w|x^n)  \Tr\left[ \rho_{Y^n}^{x^n} \tilde{T}_{Y^n}^w \right]  \right)^{c(1+\frac1t)}\nonumber \\
	&\leq  \sum_{x^n} \mu_n(x^n) \sum_{ w\in \mathcal{W}^n } P_{W^n|X^n}(w|x^n) \left( \Tr\left[ \rho_{Y^n}^{x^n}\tilde{T}_{Y^n}^w\right]  \right)^{c(1+\frac1t)}\nonumber \\
	&\leq \sum_{ w\in \mathcal{W}^n } \sum_{x^n} \mu_n(x^n)  \left(  \Tr\left[ \rho_{Y^n}^{x^n} \tilde{T}_{Y^n}^w \right]  \right)^{c(1+\frac1t)}\nonumber \\
	&\leq |\mathcal{W}^n| \sum_{x^n} \mu_n(x^n)  \left( \Tr\left[ \rho_{Y^n}^{x^n} \tilde{T}_{Y^n}^{w^\star} \right]  \right)^{ c(1+\frac1t) }\,, \label{eq:c-q1}
	\end{align}
	for some $w^*\in\mathcal{W}^n$,	where the second inequality follows from Jensen's inequality, since $x \mapsto x^{c(1+\frac1t)}$ is convex.
	Now, let
	\begin{align}
&	\Delta(\mu,\Lambda^{X\to Y}, \rho_Y,c) := \sup_{\gamma_X\ll \mu} \left\{ c\, D\left(	 \sum_x \rho_{Y}^x \,\gamma_X(x) \|\rho_Y\right) - D(\gamma_X\|\mu)
	\right\} \label{eq:d_cq} \,,
	\end{align}
	where the supremum in  ~\eqref{eq:d_cq} is taken over all probability measures $\gamma_X$ on $\mathcal{X}$ and $\Lambda $ is the c-q channel defined in \eqref{cqxy}.
	Then, we apply Proposition~\ref{prop:variational}, given below, to the right-hand side of~\eqref{eq:c-q1} to obtain
	\begin{align}
	&(1-\eps-\eps' - \delta)^{ c(1+\frac1t) } \leq 
	|\mathcal{W}^n| \e^{ \Delta(\mu_n, \Lambda^{\otimes n}, \rho_Y^{\otimes n}, c) }  \left( \Tr \left[ \rho_Y^{\otimes n} \Psi_{t}^{\otimes n}(\tilde{T}^{w^\star}_{Y^n}) \right] \right)^c, \label{eq:c-q1.5}
	\end{align}
	where, given $\rho_Y=\sum_{x\in\cX} Q_X(x)\,\rho^x_Y $ and  $\gamma:=\max_{x \in\cX}\left\| \rho_Y^x \rho_Y^{-1} \right\|_\infty$,
	\begin{align} \label{eq:Psi}
	\Psi_t(T):= \e^{-t} T + \gamma (1-\e^{-t})\Tr\left[\rho_Y T\right]\mathds{1} \quad \forall \,\, T \in \cB(\cH_Y).
	\end{align}
	
	On the other hand, one can estimate
	\begin{align}
	\left( \Tr \left[ \rho_Y^{\otimes n} \Psi_{t}^{\otimes n}(\tilde{T}^{w^\star}_{Y^n}) \right] \right)^c
	&= \left( \e^{-t}  + \gamma (1-\e^{-t}) \right)^{cn}  \left( \Tr \left[ \rho_Y^{\otimes n} \tilde{T}_{Y^n}^{w^\star} \right] \right)^c\nonumber \\
	&\leq\e^{c(\gamma-1)nt}\left( \Tr \left[ \rho_Y^{\otimes n} \tilde{T}_{Y^n}^{w^\star} \right] \right)^c, \label{eq:c-q2}
	\end{align}
where the last inequality follows from the fact that $e^{\gamma t} - 1 \geq \gamma(\e^t-1)$ for $\gamma = \max_x \left\| \rho_Y^x \rho_Y^{-1} \right\|_\infty \geq 1$ (see e.g.~the proof of Theorem 29 in \cite{BDR18}).

	Combining \eqref{eq:c-q1.5} and \eqref{eq:c-q2} yields
	\begin{align}
	&(1-\eps-\eps' - \delta)^{ c(1+\frac1t)} \notag \\
	&\leq   |\mathcal{W}^n| \e^{ \Delta(\mu_n, \Lambda^{\otimes n}, \rho_Y^{\otimes n}, c) } \e^{c(\gamma-1)nt}\left( \Tr \left[ \rho_Y^{\otimes n} \tilde{T}_{Y^n}^{w^\star} \right] \right)^c \nonumber\\
	&\leq   |\mathcal{W}^n| \e^{ \Delta(\mu_n, \Lambda^{\otimes n}, \rho_Y^{\otimes n}, c) } \e^{c(\gamma-1)nt} \left( \frac{\beta}{\eps'} \right)^c, \label{eq:c-q3}
	\end{align}
	where the last inequality \eqref{eq:c-q3} follows from the construction provided in  ~\eqref{eq:new_parameter_c-q}.
	\\~\\~
	\noindent\textbf{Step 3.}
	Let
	\begin{align}
		\Delta^\star(Q_X, \Lambda , \rho_Y,c) := \sup_{ \mathcal{N}^{X\to U}} \left\{ c I(U;Y)_\omega - I(U;X)_\omega
	\right\}\,,  \label{eq:d*_cq}
	\end{align}
	where the optimization is made over the states 
	\begin{align} \label{eq:omega}
	\omega_{UXY}=\sum_{x\in\cX}Q_X(x) |x\rangle\langle x| \otimes \sum_u P_{U|X}(u|x) |u\rangle\langle u | \otimes \rho^x_Y
	\end{align} 
	(i.e.~$\omega_{UXY}$ is a c-c-q Markov chain). We claim that there exists some set $\mathcal{C}_n \subset \mathcal{X}^n$ with $Q_{X}^{\otimes n}[\mathcal{C}_n] \geq 1-\delta$ such that
	\begin{align}
	&\Delta(\mu_n,\Lambda^{\otimes n}, \rho_Y^{\otimes n}, c) 
	\leq n \Delta^\star(Q_X, \Lambda , \rho_Y, c) + \log (\eta^{c+1}) \cdot \sqrt{  {3n\eta} \log \frac{|\mathcal{X}|}{\delta} }. \label{eq:c-q4}
	\end{align}
	We remark that  ~\eqref{eq:c-q4} can be proved by following similar idea in \cite{LHV18}, the fact that conditioning reduces entropy, and the Markovian property of $Y_i - (X_1,\ldots, X_{i-1}) - (Y_1, \ldots, Y_{i-1})$ under the memoryless c-q channel $\Lambda^{\otimes n}$.
	The proof of  ~\eqref{eq:c-q4} is deferred to Appendix~\ref{sec:single-letter}.
	
	Lastly, choose 
	\begin{align}
	\eps' = \frac{1-\eps}{2},\quad 
	\delta = \frac{1-\eps}{4}, \quad
	t = \sqrt{\frac{-\log \frac{1-\eps}{4}}{\gamma n}}.
	\end{align}
	Combining  ~\eqref{eq:c-q3} and \eqref{eq:c-q4} gives the desired  ~\eqref{eq:sc_c-q}, which completes the proof.
\end{proof}

It remains to prove Proposition~\ref{prop:variational}, given below, which we used to obtain the inequality (\ref{eq:c-q1.5}).
	
	\begin{prop} \label{prop:variational}
		For any (unnormalized) probability measure $\mu_n$ on $\mathcal{X}^n$, 
$\rho_Y \in \mathcal{D}(Y)$, $0\leq T_n \leq \mathds{1}_{Y^n}$,	$c>1$ and $t>0$, it follows that
		\begin{align} \label{eq:Key}
		 \left( \Tr \left[ \rho_Y^{\otimes n} \Psi_{t}^{\otimes n}(T_n) \right] \right)^c  \e^{\Delta(\mu_n,\Lambda^{\otimes n}, \rho_Y^{\otimes n},c)}
		\geq \sum_{x^n}\mu_n(x^n) \left( \Tr\left[ \rho_{Y^n}^{x^n} T_n \right] \right)^{c(1+\frac1t) },
		\end{align}
		where $\Psi_t$ is given by \eqref{eq:Psi}, and $\Lambda\equiv \Lambda^{X\to Y}$ is the c-q channel defined in \eqref{cqxy}.
	\end{prop}

The proof of Proposition \ref{prop:variational} makes use of the Araki-Lieb-Thirring inequality, Lemma~\ref{lemm:ALT}.


	\begin{proof}[Proof of Proposition~\ref{prop:variational}]
	The result is proved using the tensorization of the reverse hypercontractivity inequality for the depolarizing semigroup together with a variational formulation of the quantum relative entropy: 
		We firstly show that
		\begin{align} \label{eq:Key0}
		\left( \Tr \left[ \rho_Y^{\otimes n} \Psi_{t}^{\otimes n}(T_n) \right] \right)^c 
		\sum_{x^n} \mu_n(x^n) \e^{c D(\rho_{Y^n}^{x^n}\| \rho_Y^{\otimes n})}
		\geq  \sum_{x^n}\mu_n(x^n) \left( \Tr\left[ \rho_{Y^n}^{x^n} T_n \right] \right)^{c(1+\frac1t) },
		\end{align}
		and secondly claim that
		\begin{align}
		\sum_{x^n} \mu_n(x^n) \e^{c D(\rho_{Y^n}^{x^n}\| \rho_Y^{\otimes n})}
		= 	\e^{\Delta(\mu_n,\Lambda^{\otimes n}, \rho_Y^{\otimes n},c)}
		\end{align}
		to complete the proof. Define, for all $x^n \in \mathcal{X}^n$,
		\begin{align}
		\label{eq:Phi}
		\Phi_{t,x^n} := \bigotimes_{i=1}^n \Phi_{t,x_i}\,, 
		\end{align}
	where, for any $x\in\cX$:	
	\begin{align}
		\Phi_{t,x}(T) := \e^{-t} T + (1-\e^{-t}) \Tr\left[ \rho_{Y}^x T \right] \mathds{1}_Y\,.
			\end{align}
		Employing the Reverse H\"older's inequality (Lemma~\ref{lemm:RHI}) with $ A = \Gamma^{-1}_{\rho_{Y^n}^{x^n}}(\rho_Y^{\otimes n})$,   $B= \Phi_{t,x^n}(T_n)$, and $p\in(0,1]$, we obtain
		\begin{align}
		\left( \Tr \left[ \rho_Y^{\otimes n} \Phi_{t,x^n}(T_n) \right] \right)^c &\geq \left\| \Gamma^{-1}_{\rho_{Y^n}^{x^n}}(\rho_Y^{\otimes n}) \right\|_{p, \rho_{Y^n}^{x^n} }^c \left\| \Phi_{t,x^n}(T_n) \right\|_{\hat{p}, \rho_{Y^n}^{x^n} }^c. \label{eq:Key2}
		\end{align}
		Applying Araki-Lieb-Thirring inequality (Lemma~\ref{lemm:ALT}), with $A = \rho_Y^{\otimes n}$, $B = (\rho_{Y^n}^{x^n})^{(1-p)/p}$ and $r=p\in(0,\frac12]$, it holds that
		\begin{align}
		\left\| \Gamma^{-1}_{\rho_{Y^n}^{x^n}}(\rho_Y^{\otimes n}) \right\|_{p, \rho_{Y^n}^{x^n} }^c
		\geq \e^{- c D_{1-p}(\rho_{Y^n}^{x^n}\|\rho_Y^{\otimes n})}, \label{eq:Key3}
		\end{align}
		where  $D_{1-p}(A\|B) := -\frac1p \log \Tr\left[ A^p B^{1-p} \right]$.
		
		To lower bound the second term in  ~\eqref{eq:Key2}, 
		we employ the quantum reverse hypercontractivity (cf.~Theorem~\ref{lemm:RHC})  with $\tau = \rho_{Y^n}^{x^n}$, $  \hat{p} \in [-1,0)$ and any $q \in [0,1)$ satisfying $\frac{1-q}{1-\hat{p}} = \e^{-t}$ to obtain
		\begin{align}
		\left\| \Phi_{t,x^n}(T_n) \right\|_{\hat{p}, \rho_{Y^n}^{x^n} }^c
		&\geq \left\| T_n \right\|_{q, \rho_{Y^n}^{x^n} }^c \\
		&= \left( \Tr\left[ \left( (\rho_{Y^n}^{x^n})^{\frac{1}{2q}} T_n  (\rho_{Y^n}^{x^n})^{\frac{1}{2q}} \right)^q  \right] \right)^{\frac{c}{q}} \\
		&\geq \left( \Tr\left[ \rho_{Y^n}^{x^n} T_n^q  \right] \right)^{\frac{c}{q}} \label{eq:Key4}\\		
		&\geq \left( \Tr\left[  \rho_{Y^n}^{x^n} T_n    \right] \right)^{\frac{c}{q}}, \label{eq:Key5}
		\end{align}
		where we used Araki-Lieb-Thirring inequality (Lemma~\ref{lemm:ALT}) again in  ~\eqref{eq:Key4}, and the last inequality \eqref{eq:Key5} holds since $0\leq T_n \leq \mathds{1}_{Y^n}$.
		
	The superoperator $(\Psi_{t}^{\otimes n} - \Phi_{t,x^n})$, where $\Phi_{t,x^n}$ is the superoperator defined through \eqref{eq:Phi},  is positivity-preserving since $\rho_Y^x \leq \mathbb{I}_B$ for all $x\in\mathcal{X}$. This can be proved by induction in $n$, as in the proof of \cite[Theorem 29]{BDR18}). Using this fact,  we obtain the following upper bound on the right-hand side of  ~\eqref{eq:Key2}: for $c >1$,
		\begin{align}
		\left( \Tr \left[ \rho_Y^{\otimes n} \Phi_{t,x^n}(T_n) \right] \right)^c &\leq \left( \Tr \left[ \rho_Y^{\otimes n} \Psi_{t}^{\otimes n}(T_n) \right] \right)^c.\label{eq:Key6}
		\end{align}
	
		Combining  ~\eqref{eq:Key2}, \eqref{eq:Key3}, \eqref{eq:Key5}, and \eqref{eq:Key6}, taking averaging over all $x^n \in \mathcal{X}^n$ with respect to the measure $\mu_n$,  we have
		\begin{align} \label{eq:Key7}
		\left( \Tr \left[ \rho_Y^{\otimes n} \Psi_{t}^{\otimes n}(T_n) \right] \right)^c 
		\sum_{x^n} \mu_n(x^n) \e^{c D_{1-p}(\rho_{Y^n}^{x^n}\| \rho_Y^{\otimes n})}
		\geq  \sum_{x^n}\mu_n(x^n) \left( \Tr\left[ \rho_{Y^n}^{x^n} T_n \right] \right)^{\frac{c}{q} }.
		\end{align}
		Taking $p\to 0$, and $\frac{1}{q}\to\frac{1}{ 1- \e^{-t}} \leq 1+\frac1t$, 
		the above inequality leads to the  ~\eqref{eq:Key0}.
		
		Next, using the variational formula~\cite{Pet88} of the quantum relative entropy of $D(\rho_{Y^n}^{x^n}\|\rho_{Y}^{\otimes n})$, we obtain
		\begin{align}
		\sum_{x^n} \mu_n(x^n) \e^{c D(\rho_{Y^n}^{x^n}\| \rho_Y^{\otimes n})} &=
		\sum_{x^n} \mu_n(x^n) \e^{ \sup_{G_n>0} \left\{c \Tr\left[ \rho_{Y^n}^{x^n} \log G_n \right] - c \log \Tr\left[ \e^{ \log \rho_Y^{\otimes n} + \log G_n } \right] \right\}} \\
		&= \sup_{G_n>0} \sum_{x^n} \mu_n(x^n) \e^{ c \Tr\left[ \rho_{Y^n}^{x^n} \log G_n \right] - c \log \Tr\left[ \e^{ \log \rho_Y^{\otimes n} + \log G_n } \right] } \\
		&= \sup_{G_n>0} \frac{\sum_{x^n} \mu_n(x^n) \e^{ c \Tr\left[ \rho_{Y^n}^{x^n} \log G_n \right]}}{ \left( \Tr\left[ \e^{ \log \rho_Y^{\otimes n} + \log G_n } \right]\right)^c  }.
		\end{align}
		Hence,
		\begin{align}
		\log \left(\sum_{x^n} \mu_n(x^n) \e^{c D(\rho_{Y^n}^{x^n}\| \rho_Y^{\otimes n})}\right) &= \sup_{G_n>0} \log\left( \sum_{x^n} \mu_n(x^n) \e^{ c \Tr\left[ \rho_{Y^n}^{x^n} \log G_n \right]}\right)  - c \log\left( \Tr\left( \e^{ \log \rho_Y^{\otimes n} + \log G_n } \right) \right) \\
		&= \Delta(\mu_n,\Lambda^{\otimes n}, \rho_Y^{\otimes n} , c)
		\end{align}
		by invoking the variational formula for $\Delta$, Proposition~\ref{prop:vaiational}, given in Appendix~\ref{sec:variational}.
		This completes the proof.		
		
	\end{proof}

\subsection{Classical-quantum second order image size characterization method}\label{cqimagesize}

The intuition behind the proof of Theorem \ref{theo:sc_cq} can be summarized as follows: given some classical encoder $\cF_n^{X^n\to W^n}$, with corresponding conditional probability distribution $P_{W^n|X^n}$, and a test $T_{W^nY^n}:=\sum_{\omega\in\cW^n}|\omega\rangle\langle\omega|\otimes T^w_{Y^n}$ such that $\mathrm{Pr}_T \{\mathsf{H}_0|\mathsf{H}_0\}:=\mathbb{E}_{(X^n,W^n)}[\tr(\rho_{Y^n}^{X^n}\,T_{Y^n}^{W^n})]\ge 1-\eps$ for some given $\eps\in(0,1)$, how small can $\mathrm{Pr}_T \{\mathsf{H}_0|\mathsf{H}_1\}:=\mathbb{E}_{W^n}[\tr(\rho_Y^{\otimes n}\,T_{Y^n}^{W^n})]$ be made? In other words, we are interested in the following optimization problem:
\begin{align*}
\min_{0\le T_{W^nY^n}\le \mathds{1}_{W^nY^n}:\,\mathbb{E}_{(X^n,W^n)}[\tr(\rho_{Y^n}^{X^n}\,T_{Y^n}^{W^n})]\ge 1-\eps}\,\mathbb{E}_{W^n}[\tr(\rho_Y^{\otimes n}\,T_{Y^n}^{W^n})]\,.
\end{align*}
Here, $X^n$ and $W^n$ refer to the corresponding registers when added as subscripts of a state $\rho$, whereas they refer to the associated classical random variables $X^n,W^n$ when added as superscripts. Then, combining (\ref{eq:c-q1}) and (\ref{eq:c-q1.5}) , we showed that the above problem can be reduced to the one of finding a lower bound on the following quantity:
\begin{align*}
\min_{0\le T_{Y^n}\le \mathds{1}_{Y^n} :\,|\mathcal{W}^n|\,\mathbb{E}_{X^n}[\tr(\rho_{Y^n}^{X^n}T_{Y^n})]\ge 1-\eps}\,\tr(\rho_{Y}^{\otimes n}T_{Y^n})\,.
\end{align*}
Note that the tests $T_{Y^n}$ over which we optimize do not depend on the register $W^n$ any longer. Then, by Markov's inequality, such a lower bound can be found by further lower bounding
\begin{align*}
\min_{0\le T_{Y^n}\le \mathds{1}_{Y^n} :\,\mathbb{P}(\tr(\rho_{Y^n}^{X^n}T_{Y^n})\ge 1-\eps)\ge \frac{1}{|\mathcal{W}^n|}}\,\tr(\rho_{Y}^{\otimes n}T_{Y^n})\,.
\end{align*}

This is an instance of what we will call a \textit{classical-quantum image size characterization problem}. Such optimization problems are directly related to the problem of determining the achievable rate region of distributed source coding, as we will see in Section \ref{secsource}. For more information on the classical image size characterization problem in this context, we refer to \cite[Chapter 15]{CK11}.
More generally, let $\Lambda:\cX\to \cD(\cH_Y)$ be a c-q channel, with $\Lambda(x)=\rho^x_Y$, $\sigma\in\cD(\cH_Y)$ and $X^n$ a random variable corresponding to a positive measure $\mu_n$ on $\cX^n$. We are interested in lower bounding the probability $\tr(\sigma^{\otimes n} T_{Y^n})$, for a given test $0\le T_{Y^n}\le \mathds{1}_{Y^n} $ on $\cH_{Y^n}$, in terms of the probability $\mathbb{P}_{\mu_n}(\tr(\rho_{Y^n}^{X^n}T_{Y^n})\ge 1-\eps)$. Using the method of Lagrange multipliers, this amounts to finding an upper bound on 
\begin{align*}
\sup_{0\le T_{Y^n}\le  \mathds{1}_{Y^n}   }\log\big( \mathbb{P}_{\mu_n}( \tr(\rho_{Y^n}^{X^n}T_{Y^n})\ge 1-\eps    ) -c\log\tr(\sigma^{\otimes n} T_{Y^n}) \big)\,,
\end{align*}	
for a given $c>0$. Next, define the following generalizations of the quantities $\Delta$ and $\Delta^*$ defined in (\ref{eq:d_cq}), resp. defined in (\ref{eq:d*_cq}):
\begin{align}
&	\Delta(\mu_X,\Lambda, \sigma,c) := \sup_{P_X\ll \mu} \big\{ c\, D\big(	 \sum_x \rho_{Y}^{x} \,\mu_X(x) \|\sigma\big) - D(P_X\|\mu_X)
\big\} \label{eq:d_cqgen} \,,\\
&	\Delta^\star(\mu_X, \Lambda, \sigma,c) := \sup_{\mathcal{N}^{X\to U}} \left\{ c D( {\sigma}_{Y|U} \| \sigma | P_U) - D(P_{X|U} \| \mu_X |P_U )\right\} \,.\label{eq:dstar_cqgen} 
\end{align}
where for example $D( {\sigma}_{Y|U} \| \sigma | P_U) := \sum_u P_U(u) D(\sigma_{Y|U=u}\| \sigma)$,
and given the conditional distribution $P_{U|X}$ corresponding to the channel $\mathcal{N}^{X\to U}$, the conditional distribution $P_{X|U}$ is defined as follows: 
\begin{align}
P_{X|U}(x|u):=\frac{ P_{UX}(u,x)}{P_U(u)}=\frac{P_{UX}(u,x)}{\sum_{y}\,P_{U|X}(u|y)\mu_X(y)}\,,
\end{align}
and $\sigma_{Y|U}:= \sum_x P_{X|U}(x) \rho_{Y}^x$. Combining Proposition \ref{prop:variational} together with (\ref{eq:c-q4}), we find the following second order strong converse to the c-q image size characterization problem, which generalizes Theorem 4.5 and Corollary 4.7 of~
\cite{LHV18}:
\begin{theo}\label{theoimagesize}
	Let $\Lambda:\cX\to \cD(\cH_Y)$ be the classical-quantum channel which outputs state $\Lambda(x)=\rho_Y^x$ for every $x\in\cX$.
	
	(i)	For any $X^n\sim\mu_n\in\cP_+(\cX^n)$, $0\le T_{Y^n}\le  \mathds{1}_{Y^n}$ on $\cH_{Y}^{\otimes n}$, $\delta\in (0,1)$ and $c>0$:
	\begin{align}\label{imagesizepbsecondorder}
	\log \mathbb{P}_{X^n}\big( \tr(\rho_{Y^n}^{X^n} T_{Y^n} )\ge \delta  \big)&-c\log\tr(\sigma^{\otimes n}  T_{Y^n})\\
	&\le \Delta(\mu_n,\Lambda^{\otimes n},\sigma^{\otimes n},c)+2\,c\sqrt{\log\frac{1}{\delta}}\,\sqrt{n(\gamma-1)}+c\,\log\frac{1}{\delta}\,,\nonumber
	\end{align}	
	where $\gamma$ is defined in Theorem \ref{theo:sc_cq}. 
	\medskip
	
	(ii) Let now $Q_X$ be a probability distribution on $\cX$ and define $\eta$ as in Theorem \ref{theo:sc_cq}. Then, for any $\eps\in(0,1)$ and $n>3\eta\log\frac{|\cX|}{\eps}$, there exists a set $\mathcal{C}_n\subset \cX^n$ with $Q_X^{\otimes n}(\mathcal{C}_n)\ge 1-\delta$ such that, for $\mu_n:=Q_X^{\otimes n}|_{\mathcal{C}_n}$, any test $T_{Y^n}$, $c>0$ and $\delta\in(0,1)$
	\begin{align}\label{simgleletter}
	\log \mathbb{P}_{X^n}\big( \tr(\rho_{Y^n}^{X^n} T_{Y^n} )\ge \delta  \big)-c\log\tr(\sigma^{\otimes n}  T_{Y^n})\le n\Delta^*(Q_X,\Lambda,\sigma,c)+A\,\sqrt{n}+c\,\log\frac{1}{\delta}\,,
	\end{align}
	where 
	\begin{align}\label{AA}
	A:=\log(\gamma^c\eta^{c+1})\sqrt{3\eta\log\frac{|\cX|}{\eps}}+2c\,\sqrt{(\gamma-1)\log\frac{1}{\delta}}\,.
	\end{align}
\end{theo}

\begin{proof}
	The proof of (\ref{imagesizepbsecondorder}) follows simply from Proposition \ref{prop:variational} after using that $$\tr(\rho_Y^{\otimes n}  \Psi_t^{\otimes n}(T_{Y^n}) )^c\le \e^{c(\gamma-1)nt}\tr(\rho_Y^{\otimes n} T_{Y^n})\,,$$ as well as Markov's inequality, so that
	\begin{align*}
	\sum_{x^n\in\cX^n}\,\mu_n(x^n)\,(\tr(\rho^{x^n}_{Y^n}T_{Y^n}))^{c(1+\frac{1}{t})}&=\mathbb{E}_{X^n}(\tr(\rho_{Y^n}^{X^n}T_{Y^n})^{c({1+\frac{1}{t}})})\\
	&\ge \delta^{c(1+\frac{1}{t})}\mathbb{P}_{X^n}(\tr(\rho_{Y^n}^{X^n}T_{Y^n})\ge \delta)
	\end{align*}
	and optimization over $t\ge 0$. Inequality (\ref{simgleletter}) follows directly from the single-letterization of Theorem \ref{theo:single-letter_LHV18}.
	
\end{proof}	

\subsection{Source coding with classical side information at the decoder}\label{secsource}

A very similar problem to the one of distributed hypothesis testing considered in Section \ref{strongconverseQHT} is the one of classical source coding with side information at the decoder, also known as the Wyner-Ahlswede-K\"orner (WAK) problem \cite{Wyn75,AGK76,hsieh2016channel,DHW18}: Let $Y^n$ be a quantum source of corresponding Hilbert space $\cH_Y^{\otimes n}\simeq \cH_{Y^n}$, and consider a classical i.i.d. register $X^n$ modeling the available side information. Here we consider the memoryless setting, where $\rho_{X^nY^n}=\rho_{XY}^{\otimes n}=\sum_{x^n\in\cX^n}|x^n\rangle\langle x^n|\otimes \rho^{x^n}_{Y^n}$ is an i.i.d. classical-quantum state. The source $Y^n$ and side information $X^n$ are then compressed separately through encoders denoted by the quantum channel $ \mathcal{G}_{n}^{Y^n\to W_2^n}:\mathcal{B}(\cH_Y^{\otimes n})\to \mathcal{B}(\cH_{W_2^n})$, resp. the classical encoding $ \cF_{n}^{X^n\to W_1^n}$ of corresponding transition map $P_{W_1^n|X^n}$. The decoding is modeled by the map $\cD^{W_1^nW_2^n\to \hat{Y}^n}$, which can equivalently be described as a family $(\mathcal{D}^{W_2^n\to \hat{Y}^n}_{w_1})_{w_1\in\mathcal{W}_1^n}$ of quantum channels. An $(\eps,n)$-code is then defined as any tuple $(\mathcal{G}_{n}^{Y^n\to W_2^n}, \cF_{n}^{X^n\to W_1^n},\cD^{W_1^nW_2^n\to \hat{Y}^n})$ such that the average square fidelity\footnote{We note that this criterion can be replaced by a more standard average fidelity criterion without changing the first order term, since $1-F\le 1-F^2\le 2(1-F)$.}

\begin{align}\label{aversquafid}
\mathbb{E}\big[F^2(\rho_{Y^n}^{X^n},\cD_{W_1^n}^{W_2^n\to \hat{Y}^n}\circ \mathcal{G}_n^{Y^n\to W_2^n} (\rho_{Y^n}^{X^n})      )\big]\ge 1-\eps\,,
\end{align}
where we recall that the fidelity between two states $\rho,\sigma$ is defined as 
\begin{align*}
F(\rho,\sigma):=\|\sqrt{\rho}\sqrt{\sigma}\|_1^2\equiv \big(\tr(\sqrt{\sqrt{\rho}{\sigma}\sqrt{\rho}})\big)^2\,.
\end{align*}

The achievability part for this can be proved using classical joint typicality encoding \cite{el2011network} to construct $\mathcal{F}_n^{X^n\to W_1^n}$ followed by a coherent state merging protocol to construct the encoder $\mathcal{G}_n^{Y^n\to W_2^n}$ together with the decoder $\mathcal{D}^{W_1^nW_2^n\to \hat{Y}^n}$ \cite{Horodecki2005}. Its proof is provided for completeness:

\begin{theo}[Achievability]\label{firstorder}
	
	Let $\cX$ be a finite alphabet, and $Y$ a quantum system with $|Y|<\infty$.  Then, there exist encoding maps $\cF_n^{X^n\to W_1^n}$, $\mathcal{G}_n^{Y^n\to W_2^n}$ and decoder $\cD^{W_1^nW_2^n\to Y^n}$ such that the average fidelity converges to $0$ as $n\to\infty$ and for each $n$,
	\begin{align}
	\frac{\log |W_2^n|}{n}\ge \inf_{U:\,U-X-Y}\,\Big(\frac{1}{2}(H(Y)_\omega+H(Y|U)_\omega):\,I(U;\,X)_\omega\le \frac{1}{n}\log|\mathcal{W}_1^n|   \Big),
	\end{align}	
	where $\omega_{UY}:=(\mathcal{N}^{X\to U}\otimes \id_{Y})(\rho_{XY})$ and $\mathcal{N}^{X\to U}$ is a classical encoder.
\end{theo}
\begin{proof}
	Let $U$ be a random variable such that $U-X-Y$ forms a c-c-q Markov chain. Given $Q_X=\frac{\log|\mathcal{W}_1^n|}{n}$, draw $2^{nQ_X}$ i.i.d. sequences $U^n(m)$, $m\in[2^{nQ_X}]$. Then, define the encoding $\cF_{n}^{X^n\to W_1^n}$ as follows: given a sequence $x^n\in\cX^n$, find $m$ such that $(x^n,U^n(m))$ belongs to the $\eps'$-typical set 
	$$\mathcal{T}_{\eps'}^{(n)}:\left\{ (x^n,u^n):\,\Big|\frac{|\{i:\,(x_i,u_i)=(x,u)\}|}{n}-P_{XU}(x,u)\Big| \le\eps'\,P_{XU}(x,u)\,,\,\,\forall (x,u)\in\,\cX\times\mathcal{U} \right\}\,.$$
	Then, define $\cF_n^{X^n\to W_1^n}(x^n)$ as the channel of corresponding output random variable $W_1^n$ as follows: if there exist more than one message $m$ such that $(X^n,U^n(m))
	\in \mathcal{T}_{\eps'}^{(n)}$, let $M$ be the smallest such $m$. On the other hand, if no such $m$ exists, let $M=1$. This defines a classical channel of rate $Q_X\ge\frac{1}{n}\,I(U;X)_{\omega}$. This way, denoting by $W_1^n=U^n(M)$, available to the decoder, we get
	\begin{align*}
	\mathbb{P}(W_1^n=u^n)&=\mathbb{P}(U^n(M)=u^n)\\
	&=\mathbb{P}((X^n,u^n)\in\mathcal{T}_{\eps'}^{(n)})
	\,.
	\end{align*}
	Denoting by $\mathcal{N}^{X\to U}$ the classical channel of corresponding conditional probability distribution $P_{U|X}$, this in particular implies that
	\begin{align}
	\|  ((\mathcal{N}^{X\to U}\otimes \id_Y)(\rho_{XY}))^{\otimes n}   - (\cF_n^{X^n\to W_1^n}\otimes \id_{Y^n})(\rho_{X^nY^n}) \|_1&\le \|P_U^{\otimes n}-P_{W_1^n}\|_1\nonumber\\
	&=    \sum_{u^n\in\mathcal{U}^n}\,|P_U^{\otimes n}(u^n)-P_{W_1^n}(u^n)|\nonumber\\
	&= \sum_{u^n\in\mathcal{U}^n}\sum_{x^n\in\cX^n:\,(x^n,u^n)\in \mathcal{T}_{\eps'}^{(n)}}\,P_{XU}^{\otimes n}(x^n,u^n)\nonumber   \\
	&=P_{XU}^{\otimes n}((\mathcal{T}_{\eps'}^{(n)})^c)\nonumber\\
	&\underset{n\to\infty}{\to}0\,,\label{eqtypicalclass}
	\end{align}
	We then perform a fully quantum Slepian Wolf protocol \cite{horodecki2007quantum} (also known as coherent state merging protocol) between the decoder holding the system $W^n_1$ and Bob who holds subsystem $Y^n$, which allows the latter to send his system to the decoder by sending qubits at a rate $\frac{1}{2}(H(Y)_\omega+H(Y|U)_\omega)$ (see Theorem 6 of \cite{khanian2018distributed}). The result follows after minimization over all possible encoding maps $\mathcal{N}^{X\to U}$ of Alice.

	%
	%
	
\end{proof}

\begin{remark}
	The classical counterpart of Theorem \ref{firstorder} was originally proved in \cite{Wyn75,AGK76}. A fully quantum version of this result was also recently found by means of quantum reverse Shannon theorem and the fully quantum Slepian Wolf protocol in \cite{hsieh2016channel,DHW18}.
\end{remark}

In the next theorem, we establish a finite sample size strong converse bound for our c-q WAK problem via quantum reverse hypercontractivity. This result can be seen as a generalization of Theorem 4.12 of \cite{LHV18}, and is a consequence of the c-q image size characterization method introduced in Section \ref{cqimagesize}. 

\begin{theo}[Second order strong converse bound]\label{theosourcecoding}
	Let $\cX$ be a finite alphabet, and $Y$ a quantum system with $|Y|<\infty$. Let also $\eps\in(0,1)$ and $n>3\eta\log\frac{|\cX|}{\eps}$, where $\eta$ is defined as in 
	(\ref{theo:sc_cq}). Then, for any encoding maps $\cF_n^{X^n\to W_1^n}$, $\mathcal{G}_n^{Y^n\to W_2^n}$ and decoder $\cD^{W_1^nW_2^n\to Y^n}$ such that the average square fidelity criterion (\ref{aversquafid}) is satisfied, we have
	\begin{align}
	\frac{\log |W_2^n|}{n}\ge&\, \,\inf_{U:\,U-X-Y}\,\Big(H(Y|U):\,I(U;\,X)\le \frac{1}{n}\log|\mathcal{W}_1^n|   \Big)\nonumber\\
	&-\Big(2\log(|Y|\eta)\sqrt{3\eta\log\frac{4|\cX|}{1-\eps}}+  2\sqrt{|Y|\,\log\frac{2}{1-\eps}} \Big)\frac{1}{\sqrt{n}}-\frac{2\log\frac{4}{1-\eps}}{n}\,.
	\end{align}	
\end{theo}

\begin{proof}
	Consider the random state ${\rho}_{\hat{Y}^n} :=\cD_{W_1^n}^{W_2^n\to \hat{Y}^n}\circ \mathcal{G}_n^{Y^n\to W_2^n} (\rho_{Y^n}^{X^n})     $. Then, from the average fidelity condition:
	\begin{align*}
	1-\eps\le 	\mathbb{E}\big[ F^2(\rho_{Y^n}^{X^n},\rho_{\hat{Y}^n}))  \big]&\equiv\mathbb{E}\big[\tr\big( \sqrt{\sqrt{\rho_{Y^n}^{X^n}}  \,{\rho}_{\hat{Y}^n}\, \sqrt{\rho_{Y^n}^{X^n}}}  \,\big)^4 \big]\\
	&\le \mathbb{E}\big[ \tr\big( \sqrt{\rho_{Y^n}^{X^n}}\,\sqrt{{\rho}_{\hat{Y}^n}}  \big)^2 \big]\\
	&= \mathbb{E}\big[ \tr\big( \sqrt{\rho_{Y^n}^{X^n}}\,{P}_{\hat{Y}^n}\sqrt{\hat{\rho}_{Y^n}^{X^n}}  \big)^2 \big]\\
	&\le  \mathbb{E}\big[ \tr\big( \rho_{Y^n}^{X^n}\,{P}_{\hat{Y}^n}\big)\,\tr\big(
	{\rho}_{\hat{Y}^n}  \big) \big]\\
	&= \mathbb{E}\big[ \tr\big( \rho_{Y^n}^{X^n}\,{P}_{\hat{Y}^n}\big) \big]\,,
	\end{align*}
	where ${P}_{\hat{Y}^n}$ denotes the random projection onto the support of ${\rho}_{\hat{Y}^n}$. The second line above comes from the Reverse Araki-Lieb-Thirring inequality (Lemma~\ref{theoALTrev}), with $A=\rho_{\hat{Y}_n}$, $B=\rho^{X^n}_{Y^n}$, $r=\frac{1}{2}$ and $a=b=4$. The second inequality comes from an application of the Cauchy-Schwartz inequality. Then, by applying Markov's inequality to $\mathbb{P}_{X^n}(1-\tr(\rho_{Y^n}^{X^n}\,P_{\hat{Y}^n})>\eps')$, where $X^n$ is distributed according to $ Q_X^{\otimes n}$ and $\eps'\in(\eps,1)$, we get that
	\begin{align*}
	\mathbb{P}_{X^n}(\tr(\rho_{Y^n}^{X^n}{P}_{\hat{Y}^n})\ge 1-\eps')\ge 1-\frac{\eps}{\eps'}\,.
	\end{align*}	
	Next, fix $\delta'\in(0,1-\eps/\eps')$ such that $n>3\eta\log\frac{|\cX|}{\delta'}$ and $\mu_n$ as in Theorem \ref{theoimagesize}. Then for $\tilde{X}^n$ distributed according to $\mu_n$:
	\begin{align*}
	\mathbb{P}_{\tilde{X}^n}(\tr(\rho_{Y^n}^{\tilde{X}^n}{P}_{\hat{Y}^n})\ge 1-\eps')\ge 1-\frac{\eps}{\eps'}-\delta'\,.
	\end{align*}
	Since the output space $\cW_1^n$ is finite, there must exist an index $w_1^*\in\mathcal{W}_1^n$ such that
	\begin{align*}
	\mathbb{P}_{\tilde{X}^n}(\tr(\rho_{Y^n}^{\tilde{X}^n}{P}_{\hat{Y}^n}(w_1^*))\ge 1-\eps')\ge \frac{1-\frac{\eps}{\eps'}-\delta'}{|\mathcal{W}_1^n|}\,.
	\end{align*}	
	Next, let $\sigma $ be the completely mixed state on $\cH_{Y}$, so that $\tr(\sigma^{\otimes n} {P}_{\hat{Y}^n} )=|Y|^{-n} \tr({P}_{\hat{Y}^n})\le \,|W_2^n|\,|Y|^{-n}$, by the Rank Nullity Theorem applied to the decoding maps. Applying Theorem (\ref{theoimagesize}) (ii) with $T_{Y^n}={P}_{\hat{Y}^n}$ and $\delta=1-\eps'$ yields
	\begin{align*}
	&\log\frac{1-\frac{\eps'}{\eps}-\delta'}{|\mathcal{W}_1^n|}-c\,\log(|W_2^n|\,|Y|^{-n})\\
	&\le \log\,\mathbb{P}_{\tilde{X}^n}(\tr(\rho_{Y^n}^{\tilde{X}^n}{P}_{\hat{Y}^n}(x^*))\ge 1-\eps')-c\,\log\tr(\sigma^{\otimes n}\,{P}_{\hat{Y}^n}(x^*))\\
	&\le n\,\Delta^*(Q_X,\,\Lambda,\,\sigma,\,c)+A\,\sqrt{n}+c\,\log\frac{1}{1-\eps'}\,,
	\end{align*}
	where $A$ is defined in (\ref{AA}). After a simple rearrangement of the terms, we get:
	\begin{align*}
	\log |\mathcal{W}_1^n|+\,c\,\log|W_2^n|\ge&-n\,\Big(\Delta^*(Q_X,\,\Lambda,\,\sigma,\,c)-c\,\log|Y|   \Big)\\
	&-\sqrt{n}\,\Big( \log(|Y|^c\eta^{c+1})\sqrt{3\eta\,\log\frac{4\,|\cX|}{1-\eps}}+2\,c\,\sqrt{|Y|\log\frac{2}{1-\eps}}  \Big)\\
	&-c\,\log\frac{2}{1-\eps}-\log\frac{4}{1-\eps}\,,
	\end{align*}
	for $n>3\eta\log\frac{4|\cX|}{1-\eps}$, where we choose $\eps'=\frac{1+\eps}{2}$ and $\delta'=\frac{1}{2}\big(1-\frac{\eps}{\eps'}\big)$. We conclude by showing that the first order term is correct, since by definition (\ref{eq:dstar_cqgen}) and (\ref{eq:d*_cq}): 
	\begin{align*}
	\Delta^*(Q_X,\, \Lambda,\,\sigma,\,c)-c\log|Y|&= \sup_{\mathcal{N}^{X\to U}} \left\{ c D( {\sigma}_{Y|U} \| \sigma | P_U)-c\,S(\sigma) - D(P_{X|U} \| Q_X |P_U ) \right\}\\
	&=\sup_{\mathcal{N}^{X\to U}} \left\{ -c\,S(\sigma_{Y|U})  - D(P_{X|U} \| Q_X |P_U ) \right\}\\
	&=\sup_{\mathcal{N}^{X\to U}} \left\{ -c\,S(\sigma_{Y|U})+c\,S(\rho_Y)  - D(P_{X|U} \| Q_X |P_U ) \right\}-c\,S(\rho_Y)\\
	&=\sup_{\mathcal{N}^{X\to U}} \left\{ c\,  I(U;Y)_\omega -I(U;X)_\omega\right\}-c\,S(\rho_Y)\\
	&=	\Delta^\star(Q_X, \Lambda , \rho_Y,c)-S(\rho_Y)\,,
	\end{align*}
	where the second line follows from the fact that $\sigma$ is the completely mixed state on $Y$, which implies that $D(\sigma_{Y|U}\|\sigma|P_U)=-S(\sigma_{Y|U})+S(\sigma)$. 
	The state $\omega$ in the fourth line is defined in \eqref{eq:omega}.
	The result follows after optimizing over $c\ge 1$. Hence, we conclude the proof

	\begin{lemm}[Reverse Araki-Lieb-Thirring inequality, Theorem 2.1 of \cite{iten2017pretty}]\label{theoALTrev}
		Let $A$ and $B$ be non-negative operators. Then, for $r\in(0,1]$ and $a,b\in (0,\infty]$ such that $\frac{1}{2r}=\frac{1}{2}+\frac{1}{a}+\frac{1}{b}$, we have
		\begin{align}
		\tr(B^\frac{1}{2}\,A\,B^{\frac{1}{2}})^{r}\le \Big( \tr(B^{\frac{r}{2}}A^rB^{\frac{r}{2}})  \Big)^r\,\|A^{\frac{1-r}{2}}\|_a^{2r}\,\|B^{\frac{1-r}{2}}\|_b^{2r}\,.
		\end{align} 
	\end{lemm}
	
\end{proof}

\begin{remark}
	We note the existence of a gap between the achievability bound of Theorem \ref{firstorder} and the strong converse bound provided in Theorem \ref{theosourcecoding}.
\end{remark}

\begin{remark}
	The sign of the $\mathcal{O}(\sqrt{n})$ second order term found in Theorem~\ref{theosourcecoding} is reversed in the regime $\eps<1/2$. In the classical case, this issue was corrected in \cite{liu2018dispersion} where the authors combined the reverse hypercontractivity technique with the more traditional method of types to get a dispersion bound of order $-\Omega(\sqrt{n})$ for sufficiently small $\eps$. Finding an analogous result in our present setting is postponed to future work. 
\end{remark}

\section{Conclusions} \label{sec:discussions}

In this paper, we generalize the bivariate distributed hypothesis testing problem with communication constraints studied by Berger, Ahlswede, and Csisz\'ar \cite{Ber78, AC86} to the quantum setting. We first show that the Stein exponent for this problem is given by a regularized quantum relative entropy. In the special case of testing against independence, we prove that the exponent admits a single-letter formula. The proof idea for the latter comes from the operational interpretation of the quantum information bottleneck function \cite{DHW18, hsieh2016channel}. When the underlying state is a classical-quantum state, we further establish that the Stein exponent is independent of the threshold $\eps$ on the type-I error probability and obtain a second-order strong converse bound for it.
The employed technical tool is the tensorization property of quantum reverse hypercontractivity for the generalized depolarizing semigroup \cite{LHV18, BDR18}. This technique is then extended to get a strong converse bound for the task of quantum source coding with classical side information at the decoder by considering the more general problem of classical-quantum image size characterization.

\section*{Acknowledgements}
HC was supported by the Cambridge University Fellowship and the Ministry of Science and Technology Overseas Project for Post Graduate Research (Taiwan) under Grant 108-2917-I-564-042. CR is supported by the TUM University Foundation Fellowship.

\appendix
\section{Expurgation Argument} \label{sec:expurgation}
The goal of this appendix is to explain the expurgation argument referred to in Step 1 of the proof of Theorem~\ref{theo:sc_cq}.
For any test $T_{W^n T^n} := \sum_{w \in \cW^n} |w\rangle\langle w|\otimes T_{Y^n}^w$ satisfying
\begin{align}
\begin{dcases}
\mathrm{Pr}_T\{\mathsf{H}_1|\mathsf{H}_0\} \leq \eps \\
\mathrm{Pr}_T\{\mathsf{H}_0|\mathsf{H}_1\} = \beta_n(T) \equiv \beta
\end{dcases}.
\end{align}
(here, we put a subscript on the probability to highlight its dependence on $T$) and $\eps' \in (0,1-\eps)$, one can construct a new test $\tilde{T}_{W^n Y^n}$ such that 
\begin{align} \label{eq:new_text_parameter}
\begin{dcases}
\mathrm{Pr}_{\tilde{T}}\{\mathsf{H}_1|\mathsf{H}_0\} \leq \eps + \eps', \\
\mathrm{Pr}_{\tilde{T}}\{\mathsf{H}_0|\mathsf{H}_1, w\} = \beta/\eps'; \quad \forall w\in \cW^n.
\end{dcases}
\end{align}
We note that this holds not only for the c-q case but also for the case in which $X^n$ is quantum and $\mathcal{F}^{X^n \to W^n}$ is a quantum to classical map, e.g.~it is characterized by a POVM $(\Pi_{X^n}^{w})_{w\in \cW^n}$. 
In this case, it is not hard to see that the state of $\mathsf{H}_0$, i.e.~$\rho_{XY}^{\otimes n}$, after encoding  is
\begin{align}
\sigma_{W^n Y^n} &= \mathcal{F}^{X^n \to W^n}\otimes \id_{Y^n}(\rho_{XY}^{\otimes n}) = \sum_w \Pr\left\{ w \right\} |w\rangle\langle w| \otimes \sigma_{Y^n}^w, \\
\Pr\left\{ w \right\} &= \Tr\left[\rho_{X}^{\otimes n}  \Pi_{X^n}^w  \right], \\
\sigma_{Y^n}^w &= \frac{\Tr_{X^n}[  \rho_{XY}^{\otimes n} \Pi_{X^n}^w \otimes \id_{Y^n} ]}{ \Tr[ \rho_{X}^{\otimes n} \Pi_{X^n}^w ]},
\end{align}
while that of $\mathsf{H}_1$ after encoding is
\begin{align}
\sigma_{W^n Y^n} &= 
 \sum_w \Pr\{w\}|w\rangle\langle w |\otimes \rho_{Y}^{\otimes n}.
\end{align}

Now, we present the proof of \eqref{eq:new_text_parameter}. Without loss of generality, we may assume that the elements in the set $\cW^n$ are ordered in such a way that 
\begin{align}
\mathrm{Pr}_T\{\mathsf{H}_0|\mathsf{H}_1, w\}  = \Tr\left[ \rho_{Y}^{\otimes n} T_{Y^n}^w \right]
\end{align}
is increasing in $w$.
Let
\begin{align} \label{eq:w_dagger}
{w}^\dagger := \argmin \left\{ \bar{w}\in \cW^n: \sum_{ {w}> \bar{w} } \Pr\{ {w}   \} \equiv \sum_{{w}> \bar{w} } \Tr\left[ \rho_{X}^{\otimes n} \Pi_{X^n}^{{w}} \right]
\leq \eps' \right\}.
\end{align}
We define a new test $\tilde{T}_{W^n Y^n}$ that always declares $\mathsf{H}_1$ upon receiving $w> {w}^\dagger$, and coincides with $T$ otherwise, i.e.
\begin{align}
\tilde{T}_{Y^n}^w := \begin{dcases}
T_{Y^n}^w & w\leq {w}^\dagger \\
0 & w> {w}^\dagger
\end{dcases}.
\end{align}
Then,
\begin{align}
\mathrm{Pr}_{\tilde{T}}\{\mathsf{H}_1|\mathsf{H}_0\}
&= \sum_{w\in {\cW^n} } \Pr\{w\}\mathrm{Pr}_{\tilde{T}}\{\mathsf{H}_1|\mathsf{H}_0, w\}\\
&= \sum_{w\in {\cW^n} } \Pr\{w\} \Tr\left[ \sigma_{Y^n}^w (\mathds{1}_{Y^n} - \tilde{T}_{Y^n}^w) \right] \\
&= \sum_{w\leq \bar{w} } \Pr\{w\} \Tr\left[ \sigma_{Y^n}^w (\mathds{1}_{Y^n} - T_{Y^n}^w) \right] \notag \\
&\quad + \sum_{w> \bar{w}} \Pr\{w\} \Tr\left[ \sigma_{Y^n}^w \mathds{1}_{Y^n} \right]  \\
&\leq \sum_{w\leq \bar{w} } \Pr\{w\} \Tr\left[ \sigma_{Y^n}^w T_{Y^n}^w \right]  + \eps'  \\
&\leq \eps + \eps'.
\end{align}

On the other hand, the new test $\tilde{T}$ immediately yields
\begin{align}
\mathrm{Pr}_{\tilde{T}}\{\mathsf{H}_0|\mathsf{H}_1, w\} 
= 0 \leq \frac{\beta}{\eps'}, \quad \forall\, w> w^\dagger.
\end{align}
It remains to show that 
\begin{align}
\mathrm{Pr}_{\tilde{T}}\{\mathsf{H}_0|\mathsf{H}_1, w^\dagger \} 
\leq \frac{\beta}{\eps'}
\end{align}
due to the fact that $\mathrm{Pr}_{\tilde{T}}\{\mathsf{H}_0|\mathsf{H}_1, w\} = \mathrm{Pr}_{\tilde{T}}\{\mathsf{H}_0|\mathsf{H}_1, w\} $ for all $w\leq \bar{w}$, and the assumption that $\mathrm{Pr}_{\tilde{T}}\{\mathsf{H}_0|\mathsf{H}_1, w\}$ is increasing in $w$.

Let us suppose $\mathrm{Pr}_{ {T}}\{\mathsf{H}_0|\mathsf{H}_1, w^\dagger \} 
>\frac{\beta}{\eps'}$.
Note that by the choice of $w^\dagger$ given in  ~\eqref{eq:w_dagger}, we have
\begin{align}
\sum_{w>w^\dagger}\Pr\left\{ w  \right\} \leq \eps < \sum_{w\geq w^\dagger}\Pr\left\{ w  \right\}.
\end{align}
Then, 
\begin{align}
\beta &= \sum_{w \in \cW^n } \Pr\{w\} \mathrm{Pr}_{ {T}}\{\mathsf{H}_0|\mathsf{H}_1, w\}\\
&\geq \sum_{w\geq w^\dagger } \Pr\{w\}\mathrm{Pr}_{ {T}}\{\mathsf{H}_0|\mathsf{H}_1, w\} \\
&> \sum_{w\geq w^\dagger } \Pr\{w\}  \frac{\beta}{\eps'} \\
&> \beta,
\end{align} 
which leads to a contradiction. Hence, we prove our claim in  ~\eqref{eq:new_text_parameter}.

\section{Single-Letterization} \label{sec:single-letter}

In this appendix, we derive the single letterization studied in \eqref{eq:c-q4} of Step 3 in the proof of Theorem~\ref{theo:sc_cq}.
We first show that for any $c>0$,
\begin{align}\label{eq:deltadelta*}
\Delta^\star (Q_X, \Lambda^{X \to Y}, \rho_Y, c) \leq \Delta (Q_X, \Lambda^{X \to Y}, \rho_Y, c),
\end{align}
and then the single-letterization formula
\begin{align}
\Delta( Q_X^{\otimes n}|_{\mathcal{C}_n},\Lambda^{\otimes n}, \rho_Y^{\otimes n}, c) 
&\leq n \Delta^\star(Q_X, \Lambda^{X \to Y}, \rho_Y, c) + O(\sqrt{n})
\end{align}
for some set $\mathcal{C}_n \subset \mathcal{X}^n$.

For any  classical map $\mathcal{N}_{X\to U} \equiv \mathcal{N}_{X'\to U}$, let
\begin{align}
P_{X|U} := \frac{P_{UX}}{P_U} = \frac{ \mathcal{N}_{X'\to U} (Q_{XX'}) }{ \mathcal{N}_{X\to U} (Q_{X})  },
\end{align}
and a density operator on $Y$ for every $u\in\mathcal{U}$,
\begin{align}
\sigma_{Y|U= u} := \sum_{x} P_{X|U = u} (x)\rho_{Y}^{x}.
\end{align}
Then, it follows that
\begin{align*}
\quad \sum_{u\in\mathcal{U}} P_U(u) \left[ c D( \sigma_{Y|U=u} \| \rho_Y) - D(P_{X|U = u} \| Q_X) \right] &\leq \sup_{u\in\mathcal{U}} \left[ c D( \sigma_{Y|U=u} \| \rho_Y) - D(P_{X|U = u} \| Q_X) \right] \\
&= \sup_{ \{P_{X|U = u} \}_{u\in\mathcal{U}} } \left[ c D( \sigma_{Y|U=u} \| \rho_Y) - D(P_{X|U = u} \| Q_X) \right]
\end{align*}
Here, $\sigma_{Y|U= u}$ also depends on $P_{X|U = u}$.
The optimization set can hence be relaxed to the set of all measures $\tilde{P}_X$ on $X$, i.e.
\begin{align}
\quad \sup_{ \{P_{X|U = u} \}_{u\in\mathcal{U}} } \left[ c D( \sigma_{Y|U=u} \| \rho_Y) - D(P_{X|U = u} \| Q_X) \right] &\leq \sup_{  \tilde{P}_{X} } \left[ c D( \tilde{\sigma}_{Y} \| \rho_Y) - D( \tilde{P}_X \| Q_X) \right] \\
&= \Delta (Q_X, \Lambda^{X \to Y}, \rho_Y, c),
\end{align}
where we let $\tilde{\sigma}_{Y} := \sum_x \tilde{P}_X(x) \rho_Y^x$.
Since this holds for all classical maps $\mathcal{N}^{X\to U}$, we have
\begin{align}
\Delta^\star (Q_X, \Lambda^{X \to Y}, \rho_Y, c) \leq \Delta (Q_X, \Lambda^{X \to Y}, \rho_Y, c), \quad \forall c>0.
\end{align}

On the other hand, we have the following ``reverse inequality".

\begin{theo}
	[{\cite[Theorem B.1]{LHV18}}, c-q version] \label{theo:single-letter_LHV18}
	Let $|\mathcal{X}|<\infty$, $Q_X$ a measure on $\mathcal{X}$, $\rho_Y$ a density operator on $Y$, and $\Lambda^{X \to Y} : x\mapsto \rho_{Y}^x$ a c-q channel.
	Define $\eta := 1/\min_x Q_X(x)$.
	Then, for every $\delta\in(0,1)$ and $n> 3 \eta \log \frac{|\mathcal{X}|}{\delta}$, we may choose a set $\mathcal{C}_n \subseteq \mathcal{X}^n$ with $Q_{X}^{\otimes n}[\mathcal{C}_n] \geq 1-\delta$ such that
	\begin{align} \label{eq::single-letter_LHV18}
	&\quad \Delta(\mu_n,\Lambda^{\otimes n}, \rho_Y^{\otimes n}, c) \leq n \Delta^\star(Q_X, \Lambda^{X \to Y}, \rho_Y, c) + \log (\eta^{c+1}) \cdot \sqrt{  {3n\eta} \log \frac{|\mathcal{X}|}{\delta} }
	\end{align}
	for every $c>0$, where we defined $\mu_n := Q_{X}^{\otimes n}|_{\mathcal{C}_n}$.
\end{theo}
Before going into the proof, let us see why this is true.
First note that under the memoryless c-q channel $\Lambda^{\otimes n}$, one has a ``q-c-q" Markov chain $Y_i - (X_1,\ldots, X_{i-1}) - (Y_1, \ldots, Y_{i-1})$.
Then, using the chain rule of conditional entropies, data-processing inequality under partial trace, and the Markovian property, it will be shown that
$d(\mu_n,\Lambda^{\otimes n}, \rho_Y^{\otimes n}, c)  \leq n \phi(\tilde{P}_X)$ for some $\tilde{P}_X$ being a mixture of empirical measures of sequences in the support of $\mu_n$, where
\begin{align}
\phi(\tilde{P}_X) := \sup_{\mathcal{N}^{X\to U}} \left\{ c D( \tilde{\sigma}_{Y|U} \| \rho_Y | P_U) - D( \tilde{P}_{X|U} \| Q_X |P_U )\right\}.
\end{align}
Note that here, $\tilde{P}_{X|U}$ and $\tilde{\sigma}_{Y|U=u} := \sum_x \tilde{P}_{X|U=u} (x) \rho_Y^x$ are built on $\tilde{P}_X$.
To the contrary, in viewing of the definition of $\Delta^\star$ given in  ~\eqref{eq:d*_cq}, 
the corresponding states $P_{X|U}$ and $\sigma_{Y|U}:= \sum_x P_{X|U}(x) \Lambda^{X \to Y}(x)$ are built on $Q_X$, i.e.
\begin{align}
\quad\Delta^\star(Q_X, \Lambda^{X \to Y}, \rho_Y, c)\notag &:= \sup_{\mathcal{N}^{X\to U}} \left\{ c D( {\sigma}_{Y|U} \| \rho_Y | P_U) - D( {P}_{X|U} \| Q_X |P_U )\right\} \\
&\equiv \phi(Q_X).
\end{align}
In other words, $\phi(\tilde{P}_X)$ is a ``wrong'' quantity which depends on a ``wrong" input distribution.
To overcome this, we need to choose a set $\mathcal{C}_n \subseteq \mathcal{X}^n$ and measure $\mu_n|_{\mathcal{C}_n}$ such that $\tilde{P}_X \approx Q_X$ and $\phi( \tilde{P}_X) \approx \phi(Q_X)$ (i.e.~the first-order term is matched).
By a continuity property, Lemma~\ref{lemm:cont_LHV18} below, the second-order term $O(\sqrt{n})$ in  ~\eqref{eq::single-letter_LHV18} actually comes from how far $\tilde{P}_X$ is from $Q_X$.

\begin{lemm}[Continuity {\cite[Lemma B.2]{LHV18}} ] \label{lemm:cont_LHV18}
	If $\tilde{P}_X \leq (1+\epsilon) Q_X$ for some $\epsilon \in (0,1)$, the
	\begin{align}
	\phi( \tilde{P}_X) \leq \phi(Q_X) + \log (\eta^{c+1}) \cdot \epsilon.
	\end{align}
\end{lemm}

\begin{proof}[Proof of Theorem~\ref{theo:single-letter_LHV18}]
	Denote by $\hat{P}_{X^n}$ the empirical measure of $X^n$ distributed by $Q_{X}^{\otimes n}$.
	Let  $n> 3 \eta \log \frac{|\mathcal{X}|}{\delta}$ and define
	\begin{align} \label{eq:set}
	\begin{split}
	\epsilon_n &:= \sqrt{ \frac{3\eta}{n} \log \frac{|\mathcal{X}|}{\delta} } \in (0,1), \\
	\mathcal{C}_n &:= \left\{ x^n : \hat{P}_{x^n} \leq (1+\epsilon_n) Q_X \right\}.
	\end{split}
	\end{align}
	For each $x$,
	\begin{align}
	\Pr\left\{  \hat{P}_{X^n}(x) > (1+\epsilon_n) Q_X(x) 
	\right\} \leq \e^{- \frac{n}{3} Q_X(x) \epsilon_n^2 } \leq \frac{\delta}{|X|}
	\end{align}
	by Chernoff bound for Bernoulli variables (see \cite[Lemma B.3]{LHV18}).
	Then, union bound implies that $Q_{X}^{\otimes n}[\mathcal{C}_n] \geq 1-\delta$. So far, it is all classical.

	Consider any $P_{X^n} \ll \mu_n := Q_{X}^{\otimes n}|_{\mathcal{C}_n}$.
	Let
	\begin{align}
	\sigma_{Y^n} := \sum_{x^n} P_{X^n}(x^n) \rho_{Y^n}^{x^n}.
	\end{align}
	Note that
	\begin{align}
	D(\sigma_{Y^n}\| \rho_{Y}^{\otimes n})
	= - H(Y^n)_\sigma - \Tr\left[ \sigma_{Y^n} \log \rho_{Y}^{\otimes n} \right],
	\end{align}
	and, further,
	\begin{align}
	\Tr\left[ \sigma_{Y^n} \log \rho_{Y}^{\otimes n} \right]
	&= \sum_{i=1}^n \Tr\left[ \sigma_{Y_i} \log  \rho_Y \right].
	\end{align}
	Using the chain rule of conditional entropies, i.e.~$H(Y^n)_\sigma = \sum_{i=1}^n H(Y_i|Y^{i-1})_\sigma$ for $Y^{i-1} := Y_1,\ldots, Y_{i-1}$, we have
	\begin{align}
	D(\sigma_{Y^n}\| \rho_{Y}^{\otimes n}) &= 
	-\sum_{i=1}^n  \left(
	H(Y_i|Y^{i-1})_\sigma + \Tr\left[ \sigma_{Y_i} \log \rho_Y \right] \right) \\
	&\leq -\sum_{i=1}^n  \left(
	H(Y_i|Y^{i-1}, X^{i-1})_\sigma + \Tr\left[ \sigma_{Y_i} \log  \rho_Y \right] \right) \label{eq:single_inequal}\\
	&= -\sum_{i=1}^n  \left(
	H(Y_i|X^{i-1})_\sigma + \Tr\left[ \sigma_{Y_i} \log \rho_Y \right] \right) \label{eq:single_equal}\\
	&= \sum_{i=1}^n D\left( \sigma_{Y_i|X^{i-1} } \| \rho_Y | P_{X^{i-1}} \right), \label{eq:single_last}
	\end{align}
	where 
	\begin{align}
	\sigma_{Y_i|X^{i-1} = x^{i-1} } :=  \sum_{x_i} P_{X_i|X^{i-1} = x^{i-1}}(x_i)\rho_{Y_i}^{x_i}.
	\end{align}
	Inequality \eqref{eq:single_inequal} follows from the fact that conditioning reduces entropies;
	equality \eqref{eq:single_equal} is due ot the Markov chain
	$Y_i - X^{i-1} - Y^{i-1}$ under the memoryless c-q channel $\Lambda^{\otimes n}$.
	To see this, here is an example of $Y_2 - X_1 - Y_1$.
	Note that $\sigma_{Y_1 Y_2 | X_1 = x_1} 
= \sum_{x_2} P_{X_1X_2|X_1 = x_1}(x_1,x_2) \rho_{Y_1}^{x_1} \otimes \rho_{Y_2}^{x_2}
=  \rho_{Y_1}^{x_1} \otimes \tilde{\rho}_{Y_2}^{x_1}$,
where we denote by a state $\tilde{\rho}_{Y_2}^{x_1} := \sum_{x_2} P_{X_1X_2|X_1 = x_1}(x_1,x_2) \rho_{Y_2}^{x_2}$ for simplicity.
Next, we will show that $I(Y_1;Y_2|X_1)_\sigma = I(Y_1; X_1 Y_2)_\sigma - I(Y_1;X_1)_\sigma = 0$.
Using the block-diagonal structure of $X_1$ and the fact that 
$D(A\otimes C \| B\otimes C) = D(A\|B)$ for every $A,B,C\geq 0$ and $\Tr[C] = 1$, it follows that
\begin{align}
D(\sigma_{X_1Y_1Y_2}\| \sigma_{Y_1} \otimes \sigma_{X_1Y_2})
&= \sum_{x_1} \Tr\left[ P_{X_1}(x_1) \rho_{Y_1}^{x_1} \otimes \tilde{\rho}_{Y_2}^{x_1}
\left( \log P_{X_1}(x_1) \rho_{Y_1}^{x_1} \otimes \tilde{\rho}_{Y_2}^{x_1}
- \log P_{X_1}(x_1) \sigma_{Y_1} \otimes \tilde{\rho}_{Y_2}^{x_1} \right) \right] \\
&=\sum_{x_1} \Tr\left[ P_{X_1}(x_1) \rho_{Y_1}^{x_1} 
\left( \log P_{X_1}(x_1) \rho_{Y_1}^{x_1} 
- \log P_{X_1}(x_1) \sigma_{Y_1} \right) \right] \\
&= D(\sigma_{X_1Y_1}\| \sigma_{X_1} \otimes \sigma_{Y_1})
\end{align}
as desired.
The more general cases of $Y_i - X^{i-1} - Y^{i-1}$ follows similarly.
	In the last line \eqref{eq:single_last}, we have used the fact $\sum_{x^{i-1}} P_{X^{i-1}} (x^{i-1}) \sigma_{Y_i|X^{i-1} = x^{i-1} } = \sigma_{Y_i}$.

	By similar arguments, one can verify that
	\begin{align}
	D(P_{X^n}\| Q_X^{\otimes n}) &= 
	\sum_{i=1}^n D\left(P_{X_i|X^{i-1}} \| Q_X | P_{X^{i-1}}\right).
	\end{align}
	Now, we have
	\begin{align}
	\quad c D(\sigma_{Y^n}\| \rho_{Y}^{\otimes n}) - D(P_{X^n}\| Q_X^{\otimes n}) &\leq \sum_{i=1}^n c D\left( \sigma_{Y_i|X^{i-1} } \| \rho_Y | P_{X^{i-1}} \right)
	- \sum_{i=1}^n D\left(P_{X_i|X^{i-1}} \| Q_X | P_{X^{i-1}}\right) \\
	&= n \left[ c D\left( \sigma_{Y_I|I X^{I-1} } \| \rho_Y | P_{X^{I-1}} \right) 
	- D\left(P_{X_I|I X^{I-1}} \| Q_X | P_{I X^{I-1}}\right) \right] \\
	&\leq \phi(P_{X_I}),
	\end{align}
	where  we have introduced a random variable $I$ uniformly distributed on $\{1,\ldots, n\}$.
	Since $P_{X_I} = \frac1n \sum_{i=1}^n P_{X_i}$ and $P_{X^n}$ is supported on $\mathcal{C}_n$, $P_{X_I}$ is a mixture of empirical measures of sequences in $\mathcal{C}_n$.
	Then, $P_{X_I} \leq (1+\epsilon_n) Q_X$.
	
	Now, invoking Lemma~\ref{lemm:cont_LHV18}, we have
	\begin{align}
	\quad d(\mu_n,\Lambda^{\otimes n}, \rho_Y^{\otimes n}, c)  &= \sup_{P_{X^n}\ll \mu_n } \left\{ cD(\sigma_{Y^n}\| \rho_Y^{\otimes n}) - D(P_{X^n}\| \mu_n) \right\} \\
	&\leq n \phi(P_{I}) \\
	&\leq n \phi( Q_X) + n \epsilon_n \\
	&= n \Delta^\star(Q_X, \Lambda^{X \to Y}, \rho_Y, c) + \log (\eta^{c+1}) \cdot n \epsilon_n,
	\end{align}
	which completes the proof.
\end{proof}

\section{A Variational Formula} \label{sec:variational}

\begin{prop}
	[A Variational Formula for $\Delta$] \label{prop:vaiational}
	Given any positive measure $\mu$ on $\mathcal{X}$, $\Lambda^{X \to Y} : \mathcal{X} \to \mathcal{D(H)}$, positive definite operator $\nu$ on $\mathcal{H}$, and constant $c>0$, it holds that
	\begin{align} \label{eq:variational}
	\Delta(\mu, \Lambda^{X \to Y}, \nu, c) = \sup_{ T > 0 } \left\{
	\log \sum_x \mu(x) \e^{c \Tr[ \rho_Y^x \log T]} - c \log \Tr\left[ \e^{\log \nu + \log T} \right]
	\right\},
	\end{align}
	where 
	$\Delta$ is defined in \eqref{eq:d_cq}.
\end{prop}
\begin{proof}
	[Proof of Proposition~\ref{prop:vaiational}]
	Recalling  ~\eqref{eq:d_cq}, we will prove 
	\begin{align} \label{eq:variational1}
	\sup_{P_X \ll \mu } \left\{ c D(\sigma_Y\|\nu) - D(P_X\|\mu)
	\right\}
	=
	\sup_{ T> 0 } \left\{
	\log \sum_x \mu(x) \e^{c \Tr[ \rho_Y^x \log T]} - c \log \Tr\left[ \e^{\log \nu + \log T} \right]
	\right\},
	\end{align}
	where we denote by $\sigma_Y := \sum_x P_X(x) \rho_Y^x$, and the supremum on the left-hand side is taken over all probability measures on $\mathcal{X}$.
	
	We commence the proof by showing ``$\geq$" in  ~\eqref{eq:variational1}.
	For any $T > 0$, we let 
	\begin{align} \label{eq:variational2}
	P_X(x) := \frac{\mu(x) \e^{c \Tr[ \rho_Y^x \log T]}  }{ \sum_{\bar x} \mu(\bar x) \e^{c \Tr[ \rho_Y^{\bar x} \log T]} }, \quad \forall x \in \mathcal{X}.
	\end{align}
	Applying the variational formula of the quantum relative entropy $D(\sigma_Y\|\nu)$ given in Proposition~\ref{prop:variational_relative} below with $G = T$ yields
	\begin{align}
	 \log \sum_x \mu(x) &\e^{c \Tr[ \rho_Y^x \log T]} - c \log \Tr\left[ \e^{\log \nu + \log T} \right]\nonumber \\
	&\leq \log \sum_x \mu(x) \e^{c \Tr[ \rho_Y^x \log T]} + D(\sigma_Y\|\nu) - c \Tr\left[ \sigma_Y \log T \right]. \label{eq:variational3}
	\end{align}
	On the other hand, by the construction in  ~\eqref{eq:variational2}, one has
	\begin{align}
	D(P_X\| \mu) &= \sum_x P_X(x) \log \frac{P_X(x)}{\mu(x)} \\
	&= \sum_x P_X(x) c \Tr[ \rho_Y^x \log T] - \log \left( \sum_{\bar x} \mu(\bar x) \e^{c \Tr[ \rho_Y^{\bar x} \log T]} \right) \\
	&= c \Tr\left[ \sigma_Y \log T \right] - \log \left( \sum_{\bar x} \mu(\bar x) \e^{c \Tr[ \rho_Y^{\bar x} \log T]} \right). \label{eq:variational4}
	\end{align}
	Hence,  ~\eqref{eq:variational3} together with \eqref{eq:variational4} give
	\begin{align}
	\log \sum_x \mu(x) \e^{c \Tr[ \rho_Y^x \log T]} - c \log \Tr\left[ \e^{\log \nu + \log T} \right] 
	\leq c D(\sigma_Y \| \nu) - D(P_X\|\mu)
	\end{align}
	as desired.
	
	Next, we prove ``$\geq$" in  ~\eqref{eq:variational1}.
	For any probability measure $P_X \ll \mu$ on $\mathcal{X}$, we let 
	\begin{align}
	G(x) &:= \e^{c \Tr\left[ \rho_Y^x \log T\right]}, \quad \forall x \in \mathcal{X}; \\
	T &:= \e^{ \log \sigma_Y - \log \nu  } \in \mathcal{P(H)}; \\
	\sigma_Y &:= \sum_x P_X(x) \rho_Y^x \in \mathcal{D(H)}.
	\end{align}
	Here, it is not hard to see that $G(x) > 0 $ for all $x\in\mathcal{X}$.

	Now, we apply the variational formula of the classical relative entropy $D(P_X\|\mu)$ in Proposition~\ref{prop:variational_relative} again to obtain
	\begin{align}
	c D(\sigma_Y \| \nu) - D(P_X \|\mu) &\leq
	c D(\sigma_Y \| \nu) - \sum_x P_X(x) \log G(x) + \log \sum_x \mu(x) G(x) \\
	&= c D(\sigma_Y \| \nu) - \sum_x P_X(x) c \Tr\left[ \rho_Y^x \log T \right] + \log \sum_x \mu(x) \e^{c \Tr\left[ \rho_Y^x \log T\right]} \\
	&= c D(\sigma_Y \| \nu) - c \Tr\left[ \sigma_Y \log T \right] + \log \sum_x \mu(x) \e^{c \Tr\left[ \rho_Y^x \log T\right]}. \label{eq:variational5}
	\end{align}
	Moreover, by the choice of $\log T = \log \sigma_Y - \log \nu$ and noting that $\sigma_Y \in \mathcal{D(H)}$, it holds that
	\begin{align}
	c D(\sigma_Y \| \nu) -  c \Tr\left[ \sigma_Y \log T \right] & =0 \\
	&= -c \log \Tr\left[ \sigma_Y \right] \\
	&= -c \log \Tr\left[ \e^{\log \nu + \log T }\right]. \label{eq:variational6}
	\end{align}
	Therefore,  ~\eqref{eq:variational5} and \eqref{eq:variational6} lead to
	\begin{align}
	c D(\sigma_Y \| \nu) - D(P_X \|\mu) &\leq 
	\log \sum_x \mu(x) \e^{c \Tr\left[ \rho_Y^x \log T\right]} - c \log \Tr\left[ \e^{\log \nu + \log T}\right],
	\end{align}
	which completes the proof.
	
\end{proof}

\begin{prop}
	[A Variational Formula for Quantum Relative Entropy {\cite{Pet88}}] \label{prop:variational_relative}
	For any $\rho\in\mathcal{D(H)}$ and $\sigma \in \mathcal{P(H)}$ such that $\rho \ll \sigma$, it holds that
	\begin{align}
	D(\rho\|\sigma) = \sup_{G \gg \sigma} \left\{
	\Tr\left[ \rho \log G \right] - \log \Tr\left[ \e^{\log \sigma + \log G} \right]
	\right\},
	\end{align}
	where the supremum is taken over all positive semi-definite operators on $\mathcal{H}$ whose support contain that of~$\sigma$.
	
\end{prop}




\end{document}